\documentclass[12pt]{amsart}
\usepackage[latin1]{inputenc}
\usepackage{mathrsfs}
\usepackage{mathtools}
\usepackage{amsmath}
\usepackage{amsfonts}
\usepackage{amssymb}
\usepackage{graphicx}
\usepackage{fourier}
\usepackage{dutchcal}
\usepackage[colorlinks=true,citecolor=green,linkcolor=blue]{hyperref}
\usepackage{enumerate}
\usepackage{esint}
\usepackage{bm}
\usepackage{xcolor} 
\usepackage{verbatim}        
\usepackage{float}

\usepackage{subcaption}
\usepackage{epstopdf}

\usepackage{scalerel}

\usepackage[margin=1in]{geometry}

\usepackage[english]{babel}
\newtheorem{theorem}{Theorem}[section]
\newtheorem{lemma}[theorem]{Lemma}

\theoremstyle{definition}
\newtheorem{definition}[theorem]{Definition}

\newtheorem{corollary}[theorem]{Corollary}

\theoremstyle{remark}
\newtheorem{remark}[theorem]{Remark}

\newcommand{\abs}[1]{\left\lvert#1\right\rvert}
\newcommand{\pa}[1]{\left( #1 \right)}
\newcommand{\rpa}[1]{\left[ #1 \right]}
\newcommand{\br}[1]{\left\lbrace #1\right\rbrace}

\newcommand{\R}{\mathbb{R}}

\newcommand{\Z}{\mathbb{Z}}
\newcommand{\N}{\mathbb{N}}
\newcommand{\Sp}{\mathbb{S}}
\newcommand{\PP}{\mathscr{P}}
\newcommand{\X}{\mathcal{X}}
\newcommand{\wh}{\widehat}
\newcommand{\wt}{\widetilde}
\newcommand{\F}{\mathcal{F}}
\newcommand{\g}{\wh{g_N}}

\newcommand{\T}{\mathcal{T}}

\numberwithin{equation}{section}

\definecolor{bostonuniversityred}{rgb}{0.8, 0.0, 0.0}
 
\definecolor{byzantium}{rgb}{0.44, 0.16, 0.39}

\title{The emergence of order in many element systems}
\author[A. Einav]{Amit Einav}
\address{Amit Einav \hfill\break 
	School of Mathematical Sciences, Durham University, Upper Mountjoy Campus, Stockton Road DH1 3LE Durham, United Kingdom}
\email{amit.einav@durham.ac.uk}

\begin{document}
\maketitle

\begin{abstract}
	Our work is dedicated to the introduction and investigation of a new asymptotic correlation relation in the field of mean field models and limits. This new notion, order (as opposed to chaos), revolves around a tendency for self organisation in a given system and is expected to be observed in biological and societal models. Beyond the definition of this new notion, our work will show its applicability, and propagation, in the so-called Choose the Leader model.
\end{abstract}
{\tiny
	{KEYWORDS:}
	Mean field limits, asymptotic correlation, order}
\\{\indent\tiny{MSC Subject Classification:}
	82C22, 60F99, 92-10, 35Q82, 35B40}

\section{Introduction}\label{sec:intro}
\subsection{Background and the mean field limit approach}\label{subsec:background_and_MFL}
Systems that involve many elements, be it a gas of particles or a herd of animals, are ubiquitous in our day to day lives. It is no wonder, then, that we are fascinated with their investigation and try to model and investigate the phenomena that define and evolve such systems. 

Historically, we have three possible approaches to consider when dealing with such systems:\\
\textbf{Microscopic approach} in which we consider every element as an individual and find their tracjectories by solving a (most likely than not) coupled system of ODEs. This approach is the most accurate of the three, but also the most untenable due to the difficulty in solving such high number of coupled equations. \\
\textbf{Macroscopic approach} in which we ``zoom'' out, both in space and time, and investigate the resulting ``fluid''. This method gets rid of statistically insignificant phenomena which we won't see in the behaviour of the ensemble as a whole. The equations we consider in this case describe the evolution of the (physical) density of the resulting fluid.\\
\textbf{Mesoscopic approach} which combines the ``best of both worlds'' from the previous two approaches. The mesoscopic approach considers an \textit{average element} of the system and how it evolves, trying to keep the microscopic picture while considering only statistically significant phenomena.  

The mesoscopic approach was first introduced around the late 19th century during the golden age of the mathematical and physical investigation of the kinetic theory of gases. It has since outgrown its initial setting and is now used to describe a plethora of physical, chemical, biological and even societal and economical phenomena. 

While nowadays we have many tools to solve mesoscopic equations, which are usually non-linear by nature, one of the main problems we encounter when dealing with these equations is the question of their relationship to the (more established) microscopic setting. A prime example to this issue, and what is now known as Hilbert's 6th problem, is the question of whether or not one can show that the famous Boltzmann equation can be attained from the equations describing the motion of particles in a dilute gas. While a partial solution to this question was given in the 1975 work of Lanford \cite{Lanford1975}, a result that was recently revisited in the work of Gallagher, Saint-Raymond, and Texier \cite{GSRT2013}, we are still lacking a full answer. The search for an answer to this problem, however, helped pave the way to a new and extremely potent idea - the idea of mean field limits.

In his 1956 work, \cite{K1956}, Mark Kac has suggested a different approach to tackle the issue of the validity of the Boltzmann equation. Kac has proposed to provide a \textit{probabilistic justification} to it, instead of an exact derivation, by considering the evolution of an ``average'' model of a dilute gas that consists of $N$ particles which undergo binary collisions. 

Mathematically, Kac's model (or Kac's walk) is a jump process which describes the evolution of the probability density of an ensemble of particles. The symmetric probability density of the ensemble\footnote{The symmetry of the probability density is necessary and intuitive - if we are considering an average model we shouldn't be able to distinguish between the elements.},$F_N$, which is defined on $\pa{\mathbb{S}^{N-1}\pa{\sqrt{N}},d\sigma_N}$ where $d\sigma_N$ is the uniform probability measure on the $(N-1)$-dimensional sphere of radius $\sqrt{N}$, $\mathbb{S}^{N-1}\pa{\sqrt{N}}$, satisfies the so-called master equation
\begin{equation}\nonumber 
	\partial_t F_N \pa{\bm{V}_N,t}=\mathcal{L}_N F_N \left(\bm{V}_N ,t\right)= N(\mathcal{Q}-I)F_N\left(\bm{V}_N ,t\right), 
\end{equation}
where $\bm{V}_N=\pa{v_1,\dots,v_N}\in \mathbb{S}^{N-1}\pa{\sqrt{N}}$ and the collision operator, $\mathcal{Q}$, is given by
$$\mathcal{Q}F\pa{\bm{V}_N}=\frac{1}{\left(\begin{tabular}{ c}
		$N$ \\
		$2$
	\end{tabular}\right)}\sum_{i<j}\frac{1}{2\pi} \int_{0}^{2\pi}F_N\left(R_{i,j,\theta}\pa{\bm{V}_N}\right)d\theta,$$
with 
\begin{equation}\label{eq:kac-velocity}
	\pa{R_{i,j,\theta}\pa{\bm{V}_N}}_{l}=\begin{cases}
		v_l & l\not=i,j,\\
		v_i\pa{\theta}=v_{i}\cos\pa{\theta}+v_{j}\sin\pa{\theta}, & l=i,\\
		v_j\pa{\theta}=-v_{i}\sin\pa{\theta}+v_{j}\cos\pa{\theta}, & l=j.
	\end{cases}
\end{equation}

Boltzmann's equation, Kac's surmised, should arise as a limit, in some sense, of the evolution equation for the first marginal of $F_N$, $F_{N,1}$, which represents the behaviour of an average particle in the system. A simple calculation shows that 
\begin{equation}\label{eq:kac_pre_mean_field}
	\partial_t F_{N,1}(v) =\frac{1}{\pi} \int_{-\pi}^{\pi}\int_{\mathbb{R}} \left(F_{N,2}\pa{v(\theta),w(\theta)}-F_{N,2}(v,w)\right)dwd\theta ,
\end{equation}
where $v\pa{\theta}$ and $w\pa{\theta}$ are given by the same formula as that which defines $v_i\pa{\theta}$ and $v_j\pa{\theta}$ in \eqref{eq:kac-velocity}.
Equation \eqref{eq:kac_pre_mean_field} is not very surprising as we expect that the evolution of an average particle will be affected by its interaction with another particle, represented by the second marginal $F_{N,2}$. Equation \eqref{eq:kac_pre_mean_field} is not closed, and if one attempts to find the equation for $F_{N,2}$ they will find that it depends on the third marginal, $F_{N,3}$. One can continue this way and find the so-called \textit{BBGKY}\footnote{Bogoliubov-Born-Green-Kirkwood-Yvon.} hierarchy, which ends with the original master equation.

At this point in his analysis Kac introduced a truly novel idea which was inspired by the original work of Boltzmann. Kac realised that the model we discussed above didn't fully take into account the fact that the gas we are considering is dilute. The dilutness implies that we expect that any two given particles have very small chance to collide with one another and the more particles we have in the system -- the smaller the chance is. Intuitively speaking, what we expect is that as $N$ increases the \textit{particles become more and more independent}. In other word, for any fixed $k\in\N$ we have that the $k-$th marginal of $F_N$, $F_{N,k}$, which represents the behaviour of a group of $k$ random particles, will become more tensorised with respect to a limiting function, $f$, which represents the limiting behaviour of one average particle:
\begin{equation}\nonumber 
	\begin{split}
		&F_{N,1}(v_1)\underset{N\text{ large}}{\approx} f(v_1),\\
		&F_{N,2}(v_1,v_2)\underset{N\text{ large}}{\approx} f(v_1)f(v_2), \\
		&\vdots\\
		&F_{N,k}(\bm{V}_k)\underset{N\text{ large}}{\approx} f^{\otimes k}\pa{\bm{V}_k}.
	\end{split}
\end{equation}

Kac has defined the above property, which we now call \textit{(molecular) chaos} or \textit{chaoticity}, rigorously. This new notion provided Kac with the ``closure condition'' needed to take a limit in \eqref{eq:kac_pre_mean_field}. Kac has shown that his model remains chaotic if it starts as such, which is known as \textit{propagation of chaos}, and that the generating function for the evolved probability density satisfies the famous Boltzmann-Kac equation 
\begin{equation}\nonumber 
	\partial_t f(v) = \frac{1}{\pi}\int_{-\pi}^{\pi}\int_{\R} \left(f(v(\theta))f(w(\theta))-f(v_1)f(v_2) \right) dwd\theta,
\end{equation}
in the limit when $N$ goes to infinity. While Kac's original model only considered the case where the velocities of the particles in the ensemble are assumed to be one dimensional, the above has been extended to higher dimensions and more realistic models where the resulting mean field equation is precisely the Boltzmann equation (see, for instance, \cite{Mckean1967}). 

Kac's model and approach have had ramification beyond their immediate success -- ushering the field of mean field models and limits. We notice that his procedure relied on exactly two ingredients:
\begin{itemize}
	\item \textit{An average model for a system of interacting elements.} In our context this is an evolution equation for the probability density of the ensemble of elements\footnote{There are various ways to attain many elements models -- they can arise as the Liouville/master equation of the ensemble following deterministic or probabilistic trajectorial equations, or as a proposed master equation that relies on probabilistic reasoning. The former case is usually explored using the empirical measure and the trajectorial equations include interactions between \textit{all} the  elements in the system, scaled by a factor of the number of the elements, while the latter is based on ideas presented in our discussion of Kac's model where each interaction occurs only between two random elements (we refer to such models as mean field models).}.
	\item \textit{An asymptotic correlation relation.} This relation expresses the emerging phenomena we expect to get as the number of elements goes to infinity. For Kac's model this relation was chaoticity. 
\end{itemize}

The simplicity of the above approach, sometimes called the \textit{mean field limit approach}, opened the flood gate to the investigation of various many element models which, in recent decades, permeated into the realms of biology, chemical interactions and even sociology -- with examples which include swarming of animals, neural networks, and consensus amongst people (see \cite{APD2021,BFFT2012,BCC2011} as well as the review paper \cite{CD(II)2022} and references within).

It may come as a surprise that while the mean field limit approach is used in various settings, the \textit{only} asymptotic correlation used to this day is that of chaoticity. This, however, doesn't seem appropriate in many biological and societal situations where we expect more dependence than independence between the underlying elements. This suspicion has been confirmed in recent works of Carlen, Chatelin, Degond, and Wennberg \cite{CCDW2013,CDW2013} who have constructed an animal based model which, after appropriate scaling, deviates from chaoticity. The need for a different type of asymptotic correlation is the beginning of this work.


\subsection{Chaos, order, and Choose the Leader model}\label{subsec:chaos_order_CL}

We start this subsection by describing the Choose the Leader model, or CL model in short, introduced in the works of Carlen, Chatelin, Degond, and Wennberg \cite{CCDW2013,CDW2013}. This model will motivate our definition of a new asymptotic correlation relation - \textit{order}.

The CL model is, similarly to Kac's model, a velocity based pair-interaction jump process that describes the evolution of a system that revolves around a herd of animals or a biological swarm. 

The model consists of $N$ animals who move in a planar domain. The velocity of each individual is assumed to be of magnitude $1$ and as such can be considered to be an element of the circle $\mathbb{S}^1$. At a random time, given by a Poisson stream with a rate $\lambda>0$, a pair of animals is chosen at random uniformly amongst all the animals and undergoes a ``collision'': one of the animals, again chosen at random uniformly between the two, adapts its velocity to the second animal up to a small amount of ``noise''. Mathematically, this means that if the $i-$th and $j-$th animal interacted and the $j-$th animal decided to follow the $i-$th animal, we have that post collision 
$$\pa{v_i,v_j}\longrightarrow \pa{v_i,Zv_i},$$
where $Z$ is an independent random variable with values on $\mathbb{S}^1$ and a given density function $g$, and where we have used the notation $vw$ to indicate the velocity $e^{i\pa{\mathrm{Arg}(v)+\mathrm{Arg}(w)}}$, considering elements in $\mathbb{S}^1$ to be of the form $e^{i\theta}$.

Following on the above convention on $\mathbb{S}^1$ we can replace the velocity variables with their respective angle on the circle and conclude that the state space of the model is the $N$-dimensional torus, $\mathcal{T}^N=\rpa{-\pi,\pi}^N$ (with the appropriate identification of the end points of the intervals), and that the master equation of the above process, i.e. the equation for the probability density of the ensemble on $\T^N$ with respect to the underlying probability measure $\frac{d\theta_1\dots d\theta_N}{\pa{2\pi}^N}$, is given by
\begin{equation}\label{eq:master_CL}
	\begin{split}
		\partial_t F_N\pa{\theta_1,\dots,\theta_N} =& \frac{2\lambda}{N-1}\sum_{i<j}\Bigg\{\frac{g\pa{\theta_i-\theta_j}}{2}\Bigg( \rpa{F_N}_{\wt{j}}\pa{\theta_1,\dots,\wt{\theta}_{j},\dots,\theta_N}\\
		&+\rpa{F_N}_{\wt{i}}\pa{\theta_1,\dots,\wt{\theta}_{i},\dots,\theta_N}\Bigg)
		-F_N\pa{\theta_1,\dots,\theta_N}\Bigg\}.
	\end{split}
\end{equation}
with
\begin{equation}\nonumber 
	\rpa{F_N}_{\wt{j}}\pa{\theta_1,\dots,\wt{\theta}_{j},\dots,\theta_N}=\int_{-\pi}^{\pi}F_N\pa{\theta_1,\dots,\theta_N}\frac{d\theta_j}{2\pi}. 
\end{equation}
where we have used the notation $\pa{\theta_1,\dots,\wt{\theta}_{j},\dots,\theta_N}$ for the $(N-1)$-dimensional vector which is attained by removing $\theta_j$ from the original $N$-dimensional vector $\pa{\theta_1,\dots,\theta_N}$. We will continue and use this notation throughout this paper.

From the description of the CL model it seems that as times passes more meetings between the animals of the herd will happen and consequently greater overall mutual adherence will be observed. The emergence of these correlation, however, may strongly depend on the number of animals. Indeed, the more animals we have the less likely it is that any two given animals will meet -- increasing the time we'll have to wait before we see any emerging pattern.

In their two papers  \cite{CCDW2013,CDW2013} Carlen et al. have addressed this issue. They showed that chaos does propagate on every fixed time interval, but is broken when we rescale our time variable \textit{as well as the noise intensity $g$}. While seemingly odd, we shouldn't be surprised that the deviation from the adherence of the velocities may also depend on the number of the animals when we think of biological/societal settings -- it can be, for example, that the more animals we have, the more anxious they get and consequently they align themselves more closely when they meet. 

This intuitive idea of adherence motivates our upcoming definition of order (Definition \ref{def:order}) but before we move to it, and for the sake of completeness, we remind the reader the general definition of chaoticity:
\begin{definition}\label{def:chaos}
	Let $\mathcal{X}$ be a Polish space. We say that a sequence of symmetric probability measures, $\mu_N\in \PP\pa{\mathcal{X}^N}$ with $N\in\N$, is $\mu_0-$chaotic for some probability measure $\mu_0\in \PP\pa{\mathcal{X}}$ if for any $k\in\N$
	$$\Pi_k\mu_{N} \underset{N\to\infty}{\overset{\text{weak}}{\longrightarrow}}\mu_0^{\otimes k}$$
	where $\Pi_k \mu_N$ is the $k-$the marginal of $\mu_N$. The weak convergence in the above refers to convergence when integrating against bounded continuous functions. 
\end{definition}

It is worth to mention at this point that there are various notions of chaoticity. We refer the interested reader to \cite{HM2014} for more information.

Carlen, Degond and Wennberg have shown the propagation of chaos in general pair-interaction models in \cite{CDW2013}. In particular they have proved the following:
\begin{theorem}\label{thm:CDE_chaos}
	Assume that $\br{F_{N}(0)}_{N\in\N}$ is $f$-chaotic. Then for any $t>0$ the solution to the CL master equation \eqref{eq:master_CL} with initial datum $\br{F_{N}(0)}_{N\in\N}$, $\br{F_{N}(t)}_{N\in\N}$, is $f(t)$-chaotic. Moreover, $f\pa{t}$ satisfies the equation
	\begin{equation}\nonumber
		\partial_t f\pa{\theta,t} = \pa{g\ast f}\pa{\theta,t} - f\pa{\theta,t}.
	\end{equation}
\end{theorem}

As was mentioned before, the breaking of chaoticity is achieved by rescaling the time and intensity of the interaction in \eqref{eq:master_CL}. The time would naturally be rescaled by a factor of $N$ to guarantee that in a (rescaled) unit time \textit{all pairs of animals have interacted once}. The scaling of the interaction, on the other hand, is motivated from a standard scaling on the line -- restricted to $\rpa{-\pi,\pi}$:
\begin{definition}\label{def:g_eps}
	Given a symmetric probability density on $\R$ with respect to the Lebesgue measure $dx$, $g$, and a scaling parameter $\epsilon>0$ we define the rescaled and restricted probability density on $\T$ with respect to the underlying probability measure $\frac{d\theta}{2\pi}$, $g_\epsilon$, by 
	\begin{equation}\nonumber 
		g_\epsilon\pa{\theta} = \frac{1}{\epsilon\widetilde{g}_{\epsilon}}g\pa{\frac{\theta}{\epsilon}}
	\end{equation}
	where
	\begin{equation}\nonumber 
		\widetilde{g}_{\epsilon} =\frac{1}{2\pi}\int^{\frac{\pi}{\epsilon}}_{-\frac{\pi}{\epsilon}}g(x)dx.
	\end{equation}
\end{definition}

We will assume from this point onwards that the probability density of our interaction in the CL model is of the form described above and that its ``generator'', $g$, is a symmetric probability density with at least a finite third moment. 

To simplify the presentation of what is to follow we will write $f\in \PP\pa{\X,\mu}$ when $f$ is a probability density on $\X$ with respect to the underlying measure $\mu$. We will shorten the above notation and say that $f\in\PP\pa{\X}$ when $\mu$ is clear from the setting. In the remainder of our work we will consider the spaces $\T^k$ with the inherent measure $\frac{d\theta_1\dots d\theta_k}{\pa{2\pi}^k}$, where $k\in\N$. 

\smallskip

Following the time rescaling $t^\prime=\frac{t}{N}$ (which we will still denote as $t$) and allowing the interaction scaling parameter to depend on $N$, i.e. considering $\epsilon=\epsilon_N$ in Definition \ref{def:g_eps}, we attain the general rescaled CL master equation: 
\begin{equation}\label{eq:master_CL_rescaled}
	\begin{split}
		\partial_t F_N\pa{\theta_1,\dots,\theta_N} =& \frac{2\lambda N}{N-1}\sum_{i<j}\Bigg\{\frac{g_{\epsilon_N}\pa{\theta_i-\theta_j}}{2}\Bigg( \rpa{F_N}_{\wt{j}}\pa{\theta_1,\dots,\wt{\theta}_{j},\dots,\theta_N}\\
		&+\rpa{F_N}_{\wt{i}}\pa{\theta_1,\dots,\wt{\theta}_{i},\dots,\theta_N}\Bigg)
		-F_N\pa{\theta_1,\dots,\theta_N}\Bigg\}.
	\end{split}
\end{equation}

Carlen, Chatelin, Degond, and Wennberg  have shown the following in \cite{CCDW2013}:
\begin{theorem}\label{thm:CCDW_no_chaos}
	Consider the rescaled CL master equations \eqref{eq:master_CL_rescaled} with $\epsilon_N=\frac{1}{\sqrt{N}}$ and let $\br{F_{N}(t)}_{N\in\N}$ be the family of their solutions. If $\br{F_{N,k}\pa{t}}_{N\in\N}$ converges weakly to a family $\br{f_k(t)}_{k\in\N}$ when $N$ goes to infinity for any $k\in \N$ and $t>0$ then $\br{f_k(t)}_{k\in\N}$ is \textit{not} chaotic, i.e. $f_k(t) \not= f_1^{\otimes k}(t)$ for $k\geq 2$. 
\end{theorem}

From the construction of the model and the discussion above we are not too surprised by this result -- asymptotic independence is \textit{not} what we expect when the animals try to adhere to one another. What we do expect, in a sense, is that if we allow the correlation to reach their full potential then the entire herd moves \textit{in a single direction following a random leader}. This motivates the following new definition:

\begin{definition}\label{def:order}
	Let $\mathcal{X}$ be a Polish space. We say that a sequence of symmetric probability measures, $\mu_N\in \PP\pa{\mathcal{X}^N}$ with $N\in\N$, is $\mu_0-$ordered for some probability measure $\mu_0\in \PP\pa{\mathcal{X}}$ if for any $k\in\N$
	\begin{equation}\label{eq:def_of_order}
		\Pi_k\pa{d\mu_{N}}\pa{\theta_1,\dots,\theta_k} \underset{N\to\infty}{\overset{\text{weak}}{\longrightarrow}}d\mu_0\pa{\theta_1}\prod_{i=1}^{k-1} \delta_{\theta_i}\pa{\theta_{i+1}}
	\end{equation}
	where $\delta_{a}\pa{\cdot}$ is the delta measure concentrated at the point $a$. When $\mu_0$ has a density function $f$ with respect to an underlying measure on $\X$, $\mu$ (i.e. when $d\mu_0\pa{\theta} = f\pa{\theta}d\mu\pa{\theta}$), we will say that the sequence  $\br{\mu_N}_{N\in\N}$ is $f-$ordered and simplify \eqref{eq:def_of_order} by writing
	$$\Pi_k\pa{d\mu_{N}}\pa{\theta_1,\dots,\theta_k} \underset{N\to\infty}{\overset{\text{weak}}{\longrightarrow}}f\pa{\theta_1}\prod_{i=1}^{k-1} \delta_{\theta_i}\pa{\theta_{i+1}}.$$
\end{definition}


\begin{remark}\label{rem:alternative_def}
	Since 
	$$\prod_{i=1}^{k-1} \delta_{\theta_i}\pa{\theta_{i+1}}=\prod_{i=2}^{k} \delta_{\theta_1}\pa{\theta_i}$$
	we can reformulate Definition \ref{def:order} by requiring that 
	\begin{equation}\nonumber
		\Pi_k\pa{d\mu_{N}}\pa{\theta_1,\dots,\theta_k} \underset{N\to\infty}{\overset{\text{weak}}{\longrightarrow}}d\mu_0\pa{\theta_1}\prod_{i=2}^{k} \delta_{\theta_1}\pa{\theta_i}.
	\end{equation}
	This formalisation of order highlights a bit more the concentration of the limit of $\Pi_k\pa{d\mu_{N}}$ on the diagonal. Additionally, if $\X$ also has  a group operation, which we will denote by $+$, we can rewrite \eqref{eq:def_of_order} as 
	\begin{equation}\label{eq:def_order_simplified}
		\Pi_k\pa{d\mu_{N}}\pa{\theta_1,\dots,\theta_k} \underset{N\to\infty}{\overset{\text{weak}}{\longrightarrow}}d\mu_0\pa{\theta_1}\prod_{i=1}^{k-1} \delta\pa{\theta_{i+1}-\theta_i}
	\end{equation}
	where $\delta$ is the delta measure concentrated at $0$. This is the case in our setting where $\X=\T$ with the underlying measure $\frac{d\theta}{2\pi}$ and we will use this notation from this point onwards.
\end{remark}

\begin{remark}\label{rem:about_notation}
	As we are starting to mix between singular measures and probability densities we may encounter notational issues. To simplify the presentation of this work, we will keep using a density based notation with the understanding that 
	$$\int_{\T}h(\theta)\delta\pa{\theta-\varphi}\frac{d\theta}{2\pi} = h\pa{\varphi}$$
	for all appropriate measurable functions. 
\end{remark}

Much like when considering the notion of chaoticity, an immediate question one must ask is whether or not there are any ordered states. The answer to that is in the affirmative. Given a Polish space $\mathcal{X}$ and $\mu_0\in \PP\pa{\X}$ we can define the family
$$d\mu_N \pa{\theta_1,\dots,\theta_N}=d \mu_0\pa{\theta_1}\prod_{i=1}^{N-1} \delta_{\theta_i}\pa{\theta_{i+1}}\in \PP\pa{\X^N}$$
whose marginals clearly satisfy 
$$\Pi_k\pa{d\mu_N} \pa{\theta_1,\dots,\theta_k}= d\mu_0\pa{\theta_1}\prod_{i=1}^{k-1} \delta_{\theta_i}\pa{\theta_{i+1}}.$$	
This shouldn't come as a great surprise: since our notion or order speaks of an asymptotic concentration along the diagonal, choosing a family that already has this property produces an ordered state (this is, in a sense, equivalent to choosing a tensorised family of states in the chaotic setting). \\
It is worth to note that since 
$$d\mu_0\pa{\theta_1}\prod_{i=1}^{N-1} \delta_{\theta_i}\pa{\theta_{i+1}}=\frac{1}{N}\sum_{j=1}^Nd\mu_0\pa{\theta_j}\prod_{i+1\not=j} \delta_{\theta_i}\pa{\theta_{i+1}}$$
our family $\br{\mu_N}_{N\in\N}$ is indeed symmetric. 
\medskip

Our goal in this work is to explore the newly defined notion of order and show that it is the right asymptotic correlation relation for the rescaled CL model, at least when the interaction is strong enough. Moreover, we will show that this notion \textit{propagates}.

\subsection{Main results}\label{subsec:main_results}
As we've mentioned in the previous subsection, in order to see an emergence of a non-chaotic phenomenon we need to rescale the time and the intensity of the underlying interactions in the process. While the works of Carlen et al. discuss a specific choice of scaling intensity $\epsilon_N$, we have, in fact, three different possibilities. 

To see these possibilities more clearly, let us consider the evolution equation for the first marginal. A simple integration of \eqref{eq:master_CL_rescaled} together with the fact that for symmetric density functions
\begin{equation}\nonumber 
	\rpa{F_N}_{\wt{j}}\pa{\theta_1,\dots,\wt{\theta}_{j},\dots,\theta_N}=F_{N,N-1}\pa{\theta_1,\dots,\wt{\theta}_{j},\dots,\theta_N} 
\end{equation}
shows that the evolution of the $k$-th marginals, with $k=1,\dots,N$, is given by the following BBGKY hierarchy
\begin{equation}\label{eq:hierarchy}
	\begin{split}
		\partial_t &F_{N,k}\pa{\theta_1,\dots,\theta_k} =\frac{2\lambda N}{N-1}\sum_{i<j\leq k}\Bigg\{\frac{g_{\epsilon_N}\pa{\theta_i-\theta_j}}{2}\Bigg( F_{N,k-1}\pa{\theta_1,\dots,\wt{\theta}_{j},\dots,\theta_k}\\
		&+F_{N,k-1}\pa{\theta_1,\dots,\wt{\theta}_{i},\dots,\theta_k}\Bigg)
		-F_{N,k}\pa{\theta_1,\dots,\theta_k}\Bigg\}\\
		&+\frac{2\lambda N\pa{N-k}}{N-1}\sum_{i\leq k}\frac{1}{2}\Bigg\{\int_{\T}g_{\epsilon_N}\pa{\theta_i-\theta_{k+1}} 
		F_{N,k}\pa{\theta_1,\dots,\wt{\theta}_{i},\dots,\theta_{k+1}}\frac{d\theta_{k+1}}{2\pi}\\
		&-F_{N,k}\pa{\theta_1,\dots,\theta_k}\Bigg\}
	\end{split}
\end{equation}

(for proof in the non-scaled case, see \cite{CCDW2013}). When $k=1$ the above reads as
\begin{equation}\label{eq:evolution_of_first_marginal}
	\begin{aligned}
		&\partial_t F_{N,1}\pa{\theta_1,t}=\lambda N
		\pa{\int_{-\pi}^{\pi}g_{\epsilon_N}\pa{\theta_1-\theta} 
			F_{N,1}\pa{\theta,t}\frac{d\theta}{2\pi}
			-F_{N,1}\pa{\theta_1,t}}.
	\end{aligned}
\end{equation}
As our underlying space is $\mathcal{T}=[-\pi,\pi]$ and the above is clearly a PDE which involves convolution, we are motivated to use Fourier analysis and see that on the Fourier side equation \eqref{eq:evolution_of_first_marginal} can be rewritten as 
\begin{equation}\label{eq:evolution_of_first_marginal_fourier}
	\frac{d}{dt}\wh{F_{N,1}}\pa{n,t} = \lambda N \pa{\wh{g_{\epsilon_N}}\pa{n}-1}\wh{F_{N,1}}\pa{n,t},\qquad n\in\Z.
\end{equation}
The solution to \eqref{eq:evolution_of_first_marginal_fourier} is explicitly given by
\begin{equation}\label{eq:evolution_of_first_marginal_fourier_solution}
	\wh{F_{N,1}}\pa{n,t} = e^{\lambda N \pa{\wh{g_{\epsilon_N}}\pa{n}-1}t}\wh{F_{N,1}}\pa{n,0},\qquad n\in\Z.
\end{equation}
It can be shown that as long as $g$ has a finite third moment
\begin{equation}\label{eq:approx_g_N}
	\wh{g_{\epsilon_N}}\pa{n} = 1 +\frac{m_2}{2}\epsilon_N^2n^2 + O\pa{\epsilon_N^3\abs{n}^3},
\end{equation}
where $m_2= \int_{\R}x^2 g(x)dx$ which implies that 
\begin{equation}\nonumber
	\wh{F_{N,1}}\pa{n,t} = e^{-\lambda\pa{\frac{m_2}{2}\pa{N\epsilon_N^2}n^2+ O\pa{N\epsilon_N^3 n^3}}t}\wh{F_{N,1}}\pa{n,0},\qquad n\in\Z.
\end{equation}
The above gives rise to three scaling options:
\begin{enumerate}[(i)]
	\item \underline{$N\epsilon_N^2 \underset{N\to\infty}{\longrightarrow}0$.} In this case the interaction scaling is more dominant than the time scaling. This is the case where we expect correlation to form quickly and that \textit{order} will emerge.
	\item \underline{$N\epsilon_N^2=1$.} This is the case discussed in \cite{CCDW2013,CDW2013}. The scaled interaction and time  are ``balanced'' in a diffusive manner\footnote{This intuition is reinforced by the fact that in this case we find that $\wh{f_1} = \lim_{N\to\infty}\wh{F_{N,1}}$ is given by $$\wh{f_1}(n,t)=  e^{-\frac{\lambda m_2}{2}n^2 t}\wh{f_1}(n).$$}. Interestingly, in this case order, as defined in Definition \ref{def:order}, is not observed as we will show shortly. As a small remark we'd like to mention that we could have replaced the condition $N\epsilon_N^2=1$ with $N\epsilon_N^2 \underset{N\to\infty}{\longrightarrow}C$ where $0<C<\infty$.
	\item \underline{$N\epsilon_N^2 \underset{N\to\infty}{\longrightarrow}\infty$.} In this case the time scaling is more dominant than the interaction scaling and as a result we don't expect correlation to form quickly enough. We expect that chaos will prevail here.
\end{enumerate}

Our main results in this work concern themselves only with the first two cases as our goal is to veer away from chaoticity. Before we state our main theorems we'd like to note that the existence and uniqueness of solutions to \eqref{eq:master_CL} (and equivalently \eqref{eq:master_CL_rescaled}) is immediate from the form of the evolution equation(s) and the fact that the operators which govern them are linear and bounded\footnote{The CL master equation is of the form $$\partial_t F_N\pa{\bm{V}_N,t}=\frac{2\lambda}{N-1} \sum_{i<j}\pa{\mathcal{Q}_{i,j}^\ast - \mathrm{I}}F_N\pa{\bm{V}_N,t}$$
	where $\mathcal{Q}_{i,j}$ is a Markovian operator acting through the $i-$th and $j-$th components of $F_N$ alone. More information can be found in \cite{CDW2013}.}.

\begin{theorem}\label{thm:main_order}
	Let $\br{F_N(t)}_{N\in\N}$ be the family of symmetric solutions to \eqref{eq:master_CL_rescaled}. Assume in addition that $\lim_{N\to\infty}N\epsilon_N^2=0$ and that $\br{F_{N,k}\pa{0}}_{N\in\N}$ converges weakly as $N$ goes to infinity to a family $f_k \in \PP\pa{\mathcal{T}^k}$ for any $k\in\N$. 
	Then for any $t>0$ and any $k\in\N$, $\br{F_{N,k}(t)}_{N\in\N}$ converges weakly as $N$ goes to infinity to a family $f_k(t)\in \PP\pa{\T^k}$ which satisfies
	\begin{equation}\label{eq:limit_of_F_N_K_no_order}
		\begin{aligned}
			&f_k\pa{\theta_1,\dots,\theta_k,t}= e^{-\lambda k\pa{k-1}t}f_k\pa{\theta_1,\dots,\theta_k}\\
			&+2\lambda\int_{0}^t e^{-\lambda k\pa{k-1}\pa{t-s}}\pa{ \sum_{i<j\leq k}f_{k-1}\pa{\theta_1,\dots, \wt{\theta_i},\dots,\theta_k,s}\delta\pa{\theta_i-\theta_j}}ds.
		\end{aligned}
	\end{equation}
	In particular, we have that 
	$$\lim_{t\to\infty}f_{k}(\theta_1,\dots, \theta_k,t) = f_1\pa{\theta_1}\prod_{j=1}^{k-1}\delta\pa{\theta_{i+1}-\theta_i}$$
	which is an $f_1-$ordered family. Moreover, if $\br{F_N(0)}_{N\in\N}$ is $f_1-$ordered then 
	$$f_k\pa{\theta_1,\dots,\theta_k,t}=f_1\pa{\theta_1}\prod_{j=1}^{k-1}\delta\pa{\theta_{i+1}-\theta_i}$$ 
	for all $t>0$.
\end{theorem}

\begin{remark}\label{rem:about_main_order}
	The family of measures given by \eqref{eq:limit_of_F_N_K_no_order} is indeed a family of probability measures. To see that we notice that
	$$\int_{\T^k}f_k\pa{\theta_1,\dots,\theta_k,t}\frac{d\theta_1\dots d\theta_k}{\pa{2\pi}^k}=e^{-\lambda k\pa{k-1}t}\int_{\T^k}f_k\pa{\theta_1,\dots,\theta_k}\frac{d\theta_1\dots d\theta_k}{\pa{2\pi}^k}$$
	$$+2\lambda\int_{0}^{t} e^{-\lambda k\pa{k-1}\pa{t-s}}\pa{ \sum_{i<j\leq k}\int_{\T^{k-1}}f_{k-1}\pa{\theta_1,\dots, \wt{\theta_i},\dots,\theta_k,s}\frac{d\theta_1\dots d\wt{\theta_{i}}\dots d\theta_k}{\pa{2\pi}^{k-1}}}ds.$$
	Assuming by induction that $f_{k-1}\pa{t}$ is a probability measure shows that
	$$\int_{\T^k}f_k\pa{\theta_1,\dots,\theta_k,t}\frac{d\theta_1\dots d\theta_k}{\pa{2\pi}^k}=e^{-\lambda k\pa{k-1}t}+\lambda k\pa{k-1}\int_{0}^{t} e^{-\lambda k\pa{k-1}\pa{t-s}}ds=1,$$
	where we used the fact that $2\sum_{i<j\leq k}1=k(k-1)$.
\end{remark}
\begin{remark}\label{rem:generation_of_order}
	The first result of Theorem \ref{thm:main_order} tells us that \textit{no matter} which weakly converging family we start with, the limit family will become ordered as time goes to infinity. We can think about this as \textit{generation of order}. It is interesting to note that a phenomena of \textit{generation of chaos} was also observed by Lukkarinen. More information can be found in \cite{RS2023}.\\
	We would also like to point out that the second result in the Theorem \ref{thm:main_order} describes the \textit{propagation of order} in the CL model as it states that for any $k\in\N$ and  $t>0$
	$$\lim_{N\to\infty}F_{N,k}(t) = f_1\pa{\theta_1}\prod_{j=1}^{k-1}\delta\pa{\theta_{i+1}-\theta_i}.$$
	This capitalises on the fact that the interaction scaling is stronger then the time scaling, which is enough to imply a time independent ordered state for all $t>0$.
\end{remark}

\begin{theorem}\label{thm:main_in_between}
	Let $\br{F_N(t)}_{N\in\N}$ be the family of symmetric solutions to \eqref{eq:master_CL_rescaled}. Assume in addition that $N\epsilon_N^2=1$. Then $\br{F_N(t)}_{N\in\N} $  is neither chaotic nor ordered for any $t>0$.
\end{theorem}

Following on Theorem \ref{thm:main_order} we might wonder if the lack of order in this setting is resolved when we allow time to go to infinity. While the next theorem answers this question in the negative, it does show that there is hope for some sort of partial order (in terms of relative concentration on the diagonal) to appear. We will discuss this a bit more in \S\ref{sec:final}.

\begin{theorem}\label{thm:not_all_hope_is_lost}
	Let $\br{F_N(t)}_{N\in\N}$ be the family of symmetric solutions to \eqref{eq:master_CL_rescaled}. Assume in addition that $N\epsilon_N^2=1$ and that $\br{F_{N,1}(0)}_{N\in\N}$ and $\br{F_{N,2}(0)}_{N\in\N}$ converge weakly to $f_1\in \PP\pa{\T}$  and $f_2\in \PP\pa{\T^2}$ respectively. Then for all $t>0$ $\br{F_{N,1}(t)}_{N\in\N}$ and $\br{F_{N,2}(t)}_{N\in\N}$ converge to $f_1(t)\in\PP\pa{\T}$ and $f_2(t)\in \PP\pa{\T^2}$ respectively which satisfy
	\begin{equation}\label{eq:f_1_convergence}
		\lim_{t\to\infty} f_1(\theta_1,t) =1,
	\end{equation}
	and
	\begin{equation}\label{eq:f_2_convergence}
		\lim_{t\to\infty}f_2(\theta_1,\theta_2,t) = \mathcal{H}\pa{\theta_1-\theta_2}
	\end{equation}
	where 
	$$\mathcal{H}\pa{\theta}= \sum_{n\in\Z}\frac{2}{m_2 n^2+2}e^{in \theta}=1+4\sum_{n\in\N}\frac{\cos\pa{n\theta}}{m_2 n^2+2}.$$
\end{theorem}

\begin{remark}\label{rem:on_explicit_formula}
	While it is possible to find $f_1(t)$ and $f_2(t) $ (as we will see in the proof of the theorem), the focus of Theorem \ref{thm:not_all_hope_is_lost} is on the asymptotic behaviour with respect to time and consequently we elected to exclude formulae from the statement. 
\end{remark}

\begin{remark}\label{rem:on_h}
	As can be seen in the figure below
	\begin{figure}[H]
		\centering
		\includegraphics[scale=0.9]{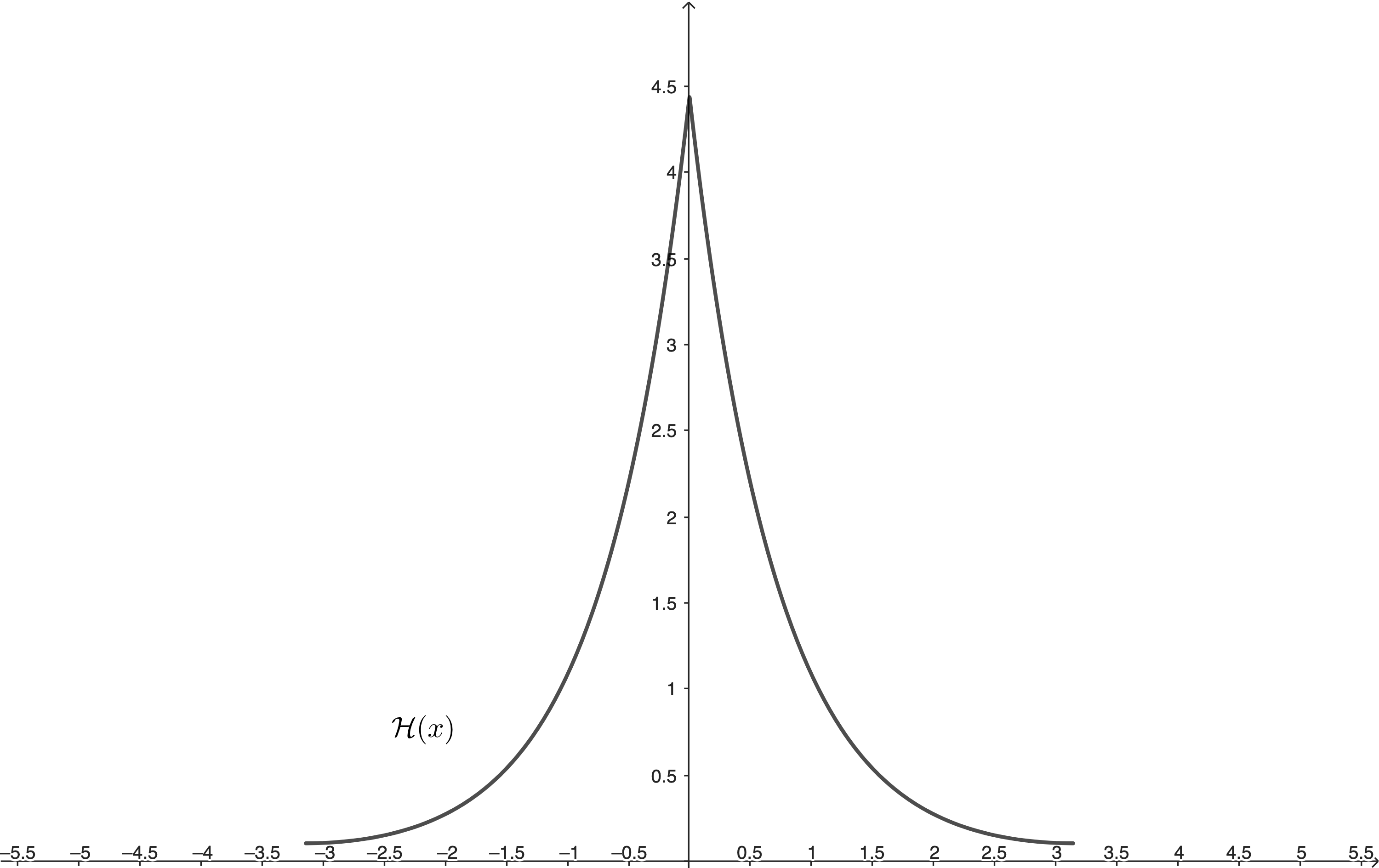}
		\caption{A plot of an approximation of $\mathcal{H}$ with $m_2=1$ by the first 500 terms of the cosine series.}
	\end{figure}
	$\mathcal{H}$ is somewhat concentrated around $0$, validating our intuition that some ``type of order'' (or partial order) phenomenon may emerge here.
\end{remark}

\subsection{The organisation of the paper} In section \S\ref{sec:preliminaries} we will discuss some preliminaries that will help us prove our main results. Section \S\ref{sec:strong_interaction} will be dedicated to the proof of Theorem \ref{thm:main_order} while section \S\ref{sec:balanced_interaction} will focus on Theorems \ref{thm:main_in_between} and \ref{thm:not_all_hope_is_lost}. We'll conclude the work with some final remarks in section \S\ref{sec:final} and a couple of appendices that consider some technical details. 

\section{Preliminaries}\label{sec:preliminaries}

Looking at the BBGKY hierarchy of our (rescaled) CL model, given by \eqref{eq:hierarchy}, we immediately notice that besides the fact that we are dealing with a \textit{closed} linear hierarchy - it also involves a simple convolution term. This motivates us to use Fourier analysis in our investigation of the model, the application of which will be the focus of this short section.

In this section we will consider the following topics: the connection between weak convergence and Fourier coefficients on $\T^k$ and the meaning of order in the Fourier space, the behaviour of the Fourier coefficients of $g_{\epsilon_N}$, and the recasting of our rescaled master equation \eqref{eq:master_CL_rescaled} in the Fourier space. 

To simplify notations we will denote by $g_N=g_{\epsilon_N}$ from this point onwards. 

We start with the following simple observation whose proof is left to Appendix \ref{app:extra} for the sake of completion.
\begin{lemma}\label{lem:weak_convergence_and_fourier}
	Let $\br{\mu^{(k)}_{N}}_{N\in\N}$ be a sequence of probability measures on $\T^k,$ and let $\mu^{(k)}\in \PP\pa{\T^k}$. Then $\mu^{(k)}_N\underset{N\to\infty}{\overset{\text{weak}}{\longrightarrow}}\mu^{(k)}$ if and only if for any $\pa{n_1,\dots,n_k}\in\Z^k$
	\begin{equation}\nonumber 
		\begin{aligned}
			\wh{\mu_{N}^{(k)}}\pa{n_1,\dots,n_k} &= \int_{\T^k} e^{-i\sum_{j=1}^k n_j \theta_j} d\mu_N^{(k)}\pa{\theta_1,\dots,\theta_k} 
			&\underset{N\to\infty}{\longrightarrow}\wh{\mu^{(k)}}\pa{n_1,\dots,n_k}. 
		\end{aligned}
	\end{equation}
\end{lemma}



We would like to remind the reader that when we consider a probability density $f_k\in \PP\pa{\T^k} $ it is always with respect to the underlying measure $\frac{d\theta_1\dots d\theta_k}{\pa{2\pi}^k}$
which means that 
$$\wh{f_k}\pa{n_1,\dots,n_k}=\int_{\T^k} f_k\pa{\theta_1,\dots,\theta_k}e^{-i\sum_{j=1}^k n_j \theta_j} \frac{d\theta_1\dots d\theta_k}{\pa{2\pi}^k},$$
as expected.
\subsection*{The meaning of order in the Fourier space} 
Following on Lemma \ref{lem:weak_convergence_and_fourier} we want to find out how an ordered state looks like in the Fourier space:

\begin{lemma}\label{lem:order_in_fourier}
	The family $F_N\in \PP\pa{\T^N}$, with $N\in\N$, is $f-$ordered if and only if 
	for any $\pa{n_1,\dots,n_k}\in\Z^k$
	\begin{equation}\nonumber 
		\begin{aligned}
			\wh{F_{N,k}}\pa{n_1,\dots,n_k}&\underset{N\to\infty}{\longrightarrow}\wh{f}\pa{\sum_{j=1}^k n_j}. 
		\end{aligned}
	\end{equation}
\end{lemma}

The proof of the above relies on the following simple observation:
\begin{lemma}\label{lem:purely_ordered_state}
	Let $f\in \PP\pa{\T}$ and let $k\in\N$. Then $\mu\in\PP\pa{\T^k}$ satisfies 
	$$d\mu\pa{\theta_1,\dots,\theta_k}=f\pa{\theta_1}\prod_{i=1}^{k-1}\delta\pa{\theta_{i+1}-\theta_i}$$
	if and only if $\wh{\mu}\pa{n_1,\dots,n_k}=\wh{f}\pa{\sum_{j=1}^k n_j}$ for any $\pa{n_1,\dots,n_k}\in \Z^k$. 
\end{lemma}


\begin{proof} 
	Since the Fourier coefficients of a measure determine it uniquely,
	it is enough for us to show that the Fourier coefficient of $\mu_{\mathrm{o}}=f\pa{\theta_1}\prod_{i=1}^{k-1}\delta\pa{\theta_{i+1}-\theta_i}$ at $\pa{n_1,\dots,n_k}$ is $\wh{f}\pa{\sum_{j=1}^{k}n_j}$. Indeed
	$$\wh{\mu_\mathrm{o}}\pa{n_1,\dots,n_k} = \int_{\T^k}f\pa{\theta_1}\prod_{i=1}^{k-1}\delta\pa{\theta_{i+1}-\theta_i}e^{-i\sum_{j=1}^k n_j\theta_j}\frac{d\theta_1\dots d\theta_k}{\pa{2\pi}^k} $$
	$$=\int_{\T}f\pa{\theta_1}e^{-i\pa{\sum_{j=1}^k n_j}\theta_1}\frac{d\theta_1}{2\pi} = \wh{f}\pa{\sum_{j=1}^k n_j}.$$
\end{proof}

\begin{proof}[Proof of Lemma \ref{lem:order_in_fourier}]
	The proof is an immediate application of Lemmas \ref{lem:weak_convergence_and_fourier} and \ref{lem:purely_ordered_state}. 
\end{proof}

\subsection*{The behaviour of the Fourier coefficients of $g_N$} 
The penultimate ingredient we need in our investigation of \eqref{eq:master_CL_rescaled} and to show the appearance of order is the following lemma:

\begin{lemma}\label{lem:approx_g}
	Let $g\in\PP\pa{\R,dx}$  be such that its $k$-th moment, defined as
	$$m_k = \int_{\R}\abs{x}^k g(x)dx,$$
	is finite for some $k>2$. Then for any $\epsilon<\frac{\pi}{\sqrt[k]{m_k}}$ and any $n\in\Z$
	\begin{equation}\label{eq:estimation_for_fourier_of_g_eps_small_xi}
		\abs{\widehat{g_\epsilon}(n) - 1 + \frac{m_2}{2}\pa{n\epsilon}^2} \leq \frac{2\epsilon^k m_k}{\pi^k-\epsilon^km_k} + \frac{m_3}{3}\pa{\abs{n}\epsilon}^3.
	\end{equation}
\end{lemma}

The proof of the above is fairly straightforward and can be found in Appendix \ref{app:fourier_behaviour} for the sake of completion.  

\begin{remark}
	We would like to point out that the approximation \eqref{eq:approx_g_N} follows immediately from the above\footnote{Note that for any $n\in\Z\setminus\br{0}$ we have that $\epsilon_N^3 \leq \epsilon_N^3 \abs{n}^3$ which shows that \eqref{eq:estimation_for_fourier_of_g_eps_small_xi} indeed implies \eqref{eq:approx_g_N} when $n\not=0$. If $n=0$ we have the identity $\g(0)=1$.}.
\end{remark}

\subsection*{Recasting of the (rescaled) master equation for the CL model in the Fourier space} The last result of this section concerns itself with recasting \eqref{eq:master_CL_rescaled} with the Fourier coefficients of our given family of solutions. We would like to mention that as the underlying space is compact and the generator of our master equation is a bounded linear operator, there is no issue with interchanging the time derivative and spatial integration which we will perform in order to move to the Fourier space. 

\begin{lemma}\label{lem:fourier_of_F_N_k}
	Let $\br{F_{N,k}(t)}_{N\in\N}$ be the family of $k$-th marginals to the family of symmetric  solutions to \eqref{eq:master_CL_rescaled}, $\br{F_N(t)}_{N\in\N}$. Then we have that 
	\begin{equation}\label{eq:fourier_of_F_N_k}
		\begin{aligned}
			&\partial_t \wh{F_{N,k}}\pa{n_1,\dots,n_k} \\
			&= \frac{2\lambda N}{\pa{N-1}}\sum_{i<j\leq k}\Bigg\{ \frac{\widehat{g_N}\pa{n_i}+\widehat{g_N}\pa{n_j}}{2}\wh{F_{N,k-1}}\pa{n_1,\dots, n_i+n_j,\dots,n_k}\\
			&- \wh{F_{N,k}}\pa{n_1,\dots,n_k}\Big\}+\frac{\lambda N\pa{N-k}}{\pa{N-1}}\wh{F_{N,k}}\pa{n_1,\dots,n_k}\sum_{i\leq k}\pa{\wh{g_N}\pa{n_i}-1},
		\end{aligned}
	\end{equation}
	where $\pa{n_1,\dots, n_i+n_j,\dots,n_k}$
	is attained by replacing $n_i$  with $n_i+n_j$ and omitting $n_j$ from the original vector $\pa{n_1,\dots,n_k}$ or, due to the symmetry of $\wh{F_{N,k-1}}$, replacing $n_j$  with $n_i+n_j$ and omitting $n_i$.
	Identity \eqref{eq:fourier_of_F_N_k} can also be rewritten as
	\begin{equation}\label{eq:fourier_of_F_N_k_nice}
		\begin{aligned}
			&\partial_t \wh{F_{N,k}}\pa{n_1,\dots,n_k}\\
			& = \frac{\lambda N}{N-1}\pa{\pa{N-k}\sum_{i\leq k}\pa{\wh{g_N}\pa{n_i}-1}-k\pa{k-1}}\wh{F_{N,k}}\pa{n_1,\dots,n_k}\\
			& +\frac{\lambda N}{N-1}\sum_{i<j\leq k} \pa{\widehat{g_N}\pa{n_i}+\widehat{g_N}\pa{n_j}}\wh{F_{N,k-1}}\pa{n_1,\dots, n_i+n_j,\dots,n_k}.
		\end{aligned}
	\end{equation}
\end{lemma}

\begin{proof}
	We start by noticing that due to the symmetry of $g$ we find that for any $i<j \leq k$
	$$\int_{\T^k}g_{N}\pa{\theta_i-\theta_j}F_{N,k-1}\pa{\theta_1,\dots,\wt{\theta}_{j},\dots,\theta_k}e^{-i\sum_{l=1}^k n_l \theta_l}\frac{d\theta_1\dots d\theta_k}{\pa{2\pi}^k}$$
	$$=\int_{\T^k}g_{N}\pa{\theta_j-\theta_i}F_{N,k-1}\pa{\theta_1,\dots,\wt{\theta}_{j},\dots,\theta_k}e^{-i\sum_{l=1}^k n_l \theta_l}\frac{d\theta_1\dots d\theta_k}{\pa{2\pi}^k}$$
	$$=\widehat{g_N}\pa{n_j}\int_{\T_{k-1}}F_{N,k-1}\pa{\theta_1,\dots,\wt{\theta}_{j},\dots,\theta_k}e^{-i\sum_{l\not=j,\;l=1}^k n_l \theta_l}e^{-in_j\theta_i}\frac{d\theta_1\dots d\wt{\theta}_j\dots d\theta_k}{\pa{2\pi}^{k-1}}$$
	$$=\widehat{g_N}\pa{n_j}\wh{F_{N,k-1}}\pa{n_1,\dots, \underbrace{n_i+n_j}_{i\text{-th position}},\dots,\wt{n_j},\dots, n_k}.$$
	Similarly 
	$$\int_{\T_k}g_{N}\pa{\theta_i-\theta_j}F_{N,k-1}\pa{\theta_1,\dots,\wt{\theta}_{i},\dots,\theta_k}e^{-i\sum_{l=1}^k n_l \theta_l}\frac{d\theta_1\dots d\theta_k}{\pa{2\pi}^k}$$
	$$=\widehat{g_N}\pa{n_i}\wh{F_{N,k-1}}\pa{n_1,\dots,\wt{n_i},\dots,\underbrace{n_i+n_j}_{j\text{-th position}},\dots,n_k}.$$
	The above implies that 
	\begin{equation}\label{eq:fourier_of_marginals_I}
		\begin{aligned}
			\mathcal{F}_{\T^k}\Bigg(& \sum_{i<j\leq k}\Bigg\{\frac{g_{N}\pa{\cdot_i-\cdot_j}}{2} 
			\Big( F_{N,k-1}\pa{\cdot_1,\dots,\wt{\cdot}_{j},\dots,\cdot_k}\\
			&+F_{N,k-1}\pa{\cdot_1,\dots,\wt{\cdot}_{i},\dots,\cdot_k}\Big)
			-F_{N,k}\pa{\cdot_1,\dots,\cdot_k}\Bigg\}\Bigg)\pa{n_1,\dots,n_k}\\
			&= \sum_{i<j\leq k}\Bigg\{ \frac{\widehat{g_N}\pa{n_j}}{2}\wh{F_{N,k-1}}\pa{n_1,\dots, \underbrace{n_i+n_j}_{i\text{-th position}},\dots,n_k}\\
			&+\frac{\widehat{g_N}\pa{n_i}}{2}\wh{F_{N,k-1}}\pa{n_1,\dots,\underbrace{n_i+n_j}_{j\text{-th position}},\dots,n_k}\\
			& - \wh{F_{N,k}}\pa{n_1,\dots,n_k}\Big\}.
		\end{aligned}
	\end{equation}
	where we used the notation of $\F_{\T^l}\pa{f}\pa{n_1,\dots,n_l}=\wh{f}\pa{n_1,\dots,n_l}$ when $f\in \PP\pa{\T^l}$. 
	Next, due to the symmetry of $F_N$, we see that for any $i\leq k$
	\begin{equation}\nonumber 
		\begin{aligned}
			&\int_{\T^k}\pa{\int_{\T}g_{N}\pa{\theta_i-\theta_{k+1}} 
				F_{N,k}\pa{\theta_1,\dots,\wt{\theta}_{i},\dots,\theta_{k+1}}\frac{d\theta_{k+1}}{2\pi}}e^{-i\sum_{l=1}^k n_l \theta_l}\frac{d\theta_1\dots d\theta_k}{\pa{2\pi}^k}\\
			& = \wh{g_N}\pa{n_i}\int_{\T_k} 
			F_{N,k}\pa{\theta_1,\dots,\wt{\theta}_{i},\dots,\theta_{k+1}}e^{-i\sum_{l\not=i,\;l=1}^k n_l \theta_l}e^{-in_i\theta_{k+1}}\frac{d\theta_1\dots d\wt{\theta}_i\dots d\theta_{k+1}}{\pa{2\pi}^k}\\
			&=\wh{g_N}\pa{n_i}\wh{F_{N,k}}\pa{n_1,\dots,n_k},
		\end{aligned}
	\end{equation}
	and consequently
	\begin{equation}\label{eq:fourier_of_marginals_II}
		\begin{aligned}
			&\mathcal{F}_{\T^k}\Bigg( \sum_{i\leq k}\Bigg\{\int_{-\pi}^{\pi}g_{N}\pa{\cdot_i-\theta_{k+1}} 
			F_{N,k}\pa{\cdot_1,\dots,\wt{\cdot}_{i},\dots,\theta_{k+1}}\frac{d\theta_{k+1}}{2\pi}
			\\
			&-F_{N,k}\pa{\cdot_1,\dots,\cdot_k}\Bigg\}\pa{n_1,\dots,n_k}
			=\wh{F_{N,k}}\pa{n_1,\dots,n_k}\sum_{i\leq k}\pa{\wh{g_N}\pa{n_i}-1}.
		\end{aligned}
	\end{equation}
	Combining \eqref{eq:fourier_of_marginals_I} and \eqref{eq:fourier_of_marginals_II} with the BBGKY hierarchy \eqref{eq:hierarchy} yields
	\begin{equation}\nonumber
		\begin{aligned}
			&\partial_t \wh{F_{N,k}}\pa{n_1,\dots,n_k} = \frac{2\lambda N}{N-1}\sum_{i<j\leq k}\Bigg\{ \frac{\widehat{g_N}\pa{n_j}}{2}\wh{F_{N,k-1}}\pa{n_1,\dots, \underbrace{n_i+n_j}_{i\text{-th position}},\dots,n_k}\\
			&+\frac{\widehat{g_N}\pa{n_i}}{2}\wh{F_{N,k-1}}\pa{n_1,\dots,\underbrace{n_i+n_j}_{j\text{-th position}},\dots,n_k} - \wh{F_{N,k}}\pa{n_1,\dots,n_k}\Big\}.\\
			&+\frac{\lambda N\pa{N-k}}{N-1}\wh{F_{N,k}}\pa{n_1,\dots,n_k}\sum_{i\leq k}\pa{\wh{g_N}\pa{n_i}-1}.
		\end{aligned}
	\end{equation}
	Since the fact that $f$ is symmetric implies that so is $\wh{f}$ (see Appendix \ref{app:extra}) 
	we conclude \eqref{eq:fourier_of_F_N_k}.
	
	To attain \eqref{eq:fourier_of_F_N_k_nice} we notice that 
	$$2\sum_{i<j \leq k} 1= k(k-1)$$
	and rearrange \eqref{eq:fourier_of_F_N_k} .
\end{proof}
An immediate corollary of the above is the following:

\begin{corollary}\label{cor:recursive}
	A recursive formula for the $k$-th marginals $\br{F_{N,k}}_{N\in\N}$ is given by
	\begin{equation}\label{eq:recursive}
		\begin{aligned}
			&\wh{F_{N,k}}\pa{n_1,\dots,n_k,t} = e^{-\frac{\lambda N}{N-1}\pa{\pa{N-k}\sum_{l\leq k}\pa{1-\wh{g_N}\pa{n_l}}+k\pa{k-1}}t}\wh{F_{N,k}}\pa{n_1,\dots,n_k,0}\\ 		
			&+\frac{\lambda N}{N-1}\sum_{i<j\leq k}\pa{\widehat{g_N}\pa{n_i}+\widehat{g_N}\pa{n_j}}\int_{0}^t e^{-\frac{\lambda N}{N-1}\pa{\pa{N-k}\sum_{l\leq k}\pa{1-\wh{g_N}\pa{n_l}}+k\pa{k-1}}\pa{t-s}} \\
			&\qquad\qquad\qquad\qquad\qquad\qquad\qquad\qquad\wh{F_{N,k-1}}\pa{n_1,\dots, n_i+n_j,\dots,n_k,s}ds.
		\end{aligned}
	\end{equation}
	Consequently
	\begin{equation}\label{eq:expression_for_fourier_F_N_2}
		\begin{aligned}
			&\wh{F_{N,2}}\pa{n_1,n_2,t} \\
			&=e^{-\frac{\lambda N}{N-1}\pa{\pa{N-2}\sum_{l\leq 2}\pa{1-\wh{g_N}\pa{n_l}}+2}t}\wh{F_{N,2}}\pa{n_1,n_2,0}+
			\pa{\widehat{g_N}\pa{n_1}+\widehat{g_N}\pa{n_2}} \\
			&\frac{e^{-\frac{\lambda N}{N-1}\pa{\pa{N-2}\sum_{l\leq 2}\pa{1-\wh{g_N}\pa{n_l}}+2}t}- e^{-\lambda N\pa{1-\wh{g_N}\pa{n_1+n_2}}t}}{\pa{N-1}\pa{1-\wh{g_N}\pa{n_1+n_2}}-\pa{N-2}\sum_{l\leq 2}\pa{1-\wh{g_N}\pa{n_l}}-2}\wh{F_{N,1}}\pa{n_1+n_2,0},	
		\end{aligned}
	\end{equation}
	where we define $\frac{e^{\alpha t}-e^{\beta t}}{\alpha-\beta}$ to be $te^{\alpha t}$ if $\alpha=\beta$. 
\end{corollary}

\begin{proof}
	\eqref{eq:recursive} is a simple ODE solution to \eqref{eq:fourier_of_F_N_k_nice}. Plugging the solution for the case $k=1$ (which is given by \eqref{eq:evolution_of_first_marginal_fourier_solution}) in the identity for $k=2$ gives
	\begin{equation}\nonumber
		\begin{aligned}
			&\wh{F_{N,2}}\pa{n_1,n_2,t} = e^{-\frac{\lambda N}{N-1}\pa{\pa{N-2}\sum_{l\leq 2}\pa{1-\wh{g_N}\pa{n_l}}+2}t}\wh{F_{N,2}}\pa{n_1,n_2,0}+ 		
			\\
			&\frac{\lambda N}{N-1}\pa{\widehat{g_N}\pa{n_1}+\widehat{g_N}\pa{n_2}}
			\int_{0}^t e^{-\frac{\lambda N}{N-1}\pa{\pa{N-2}\sum_{l\leq 2}\pa{1-\wh{g_N}\pa{n_l}}+2}\pa{t-s}} 
			\wh{F_{N,1}}\pa{n_1+n_2,s}ds \\
			&=e^{-\frac{\lambda N}{N-1}\pa{\pa{N-2}\sum_{l\leq 2}\pa{1-\wh{g_N}\pa{n_l}}+2}t}\wh{F_{N,2}}\pa{n_1,n_2,0} 		
			+\frac{\lambda N}{N-1}\pa{\widehat{g_N}\pa{n_1}+\widehat{g_N}\pa{n_2}}\\
			&\pa{\int_{0}^t e^{-\frac{\lambda N}{N-1}\pa{\pa{N-2}\sum_{l\leq 2}\pa{1-\wh{g_N}\pa{n_l}}+2}\pa{t-s}} 
				e^{\lambda N \pa{\wh{g_{N}}\pa{n_1+n_2}-1}s}ds}\wh{F_{N,1}}\pa{n_1+n_2,0}. 
		\end{aligned}
	\end{equation}
	
	Using the fact that
	\begin{equation}\label{eq:useful_exponential}
		\int_{0}^{t}e^{\alpha\pa{t-s}}e^{\beta s}ds = \frac{e^{\alpha t}-e^{\beta t}}{\alpha-\beta},
	\end{equation}
	with the convention that was mentioned in the statement of the corollary, we conclude \eqref{eq:expression_for_fourier_F_N_2}.
\end{proof}

With this in hand, we are ready to show our main theorems.

\section{The case of strong interactions}\label{sec:strong_interaction}

In this section we will show the emergence of order, and its propagation, in the case of strong interaction in the CL model. We start by noticing that Corollary \ref{cor:recursive} in \S\ref{sec:preliminaries} gives us an inkling to why Theorem \ref{thm:main_order} holds. Indeed, under the assumption that $\lim_{N\to\infty}N\epsilon_N^2=0$  we can show that
$$\lim_{N\to\infty} N\pa{1-\g(n)} =0$$
for any fixed $n$ and consequently, using \eqref{eq:expression_for_fourier_F_N_2}, we see that as long as $F_{N,1}(0)$ and $F_{N,2}(0)$ converge weakly to $f_1$ and $f_2$ respectively we have that
$$\wh{f_2}\pa{n_1,n_2,t}=\lim_{N\to\infty}\wh{F_{N,2}}\pa{n_1,n_2,t} = e^{-2\lambda t} \wh{f}_2\pa{n_1,n_2} + \pa{1-e^{-2\lambda t}}\wh{f_1}\pa{n_1+n_2} $$
$$=\F_{\T^2}\pa{e^{-2\lambda t}f_2\pa{\cdot_1,\cdot_2}+\pa{1-e^{-2\lambda t}}f_1\pa{\cdot_1}\delta\pa{\cdot_2-\cdot_1}}\pa{n_1,n_2}.$$
In other words
$$f_2\pa{\theta_1,\theta_2,t}=e^{-2t}f_2\pa{\theta_1,\theta_2} + \pa{1-e^{-2t}}f_1\pa{\theta_1}\delta\pa{\theta_2-\theta_1}$$
which fits the statements of Theorem \ref{thm:main_order}. Let us show the proof in the general case:

\begin{proof}[Proof of Theorem \ref{thm:main_order}]
	Using Lemma \ref{lem:approx_g}, we find that
	$$\abs{\g (n) -1 - \frac{m_2}{2}\epsilon_N^2 n^2} \leq C \epsilon_N^3 \abs{n}^3$$
	for all $n\in\Z$. Thus, if $\lim_{N\to\infty}N\epsilon_N^2=0$ we have that
	$$0\leq N\pa{1-\g(n)} \leq \pa{\frac{m_2}{2}n^2+C \epsilon_N \abs{n}^3}N\epsilon_N^2, $$
	where we used the fact that the Fourier coefficient of any real and symmetric probability density is always real and bounded in absolute value by $1$. \\
	We conclude from the above that for any $n\in\Z$ we have that 
	$$\lim_{N\to\infty} N\pa{1-\g(n)} =0.$$
	Next, we recall that Lemma \ref{lem:weak_convergence_and_fourier} assures us that for any $\pa{n_1,\dots,n_k}\in\Z^k$ we have that 
	$$\lim_{N\to\infty}\wh{F_{N,k}(0)}\pa{n_1,\dots,n_k}=\wh{f_k}\pa{n_1,\dots,n_k}.$$
	Moreover, since the Fourier coefficients of any probability measure are bounded uniformly by $1$, we can apply the Dominated Convergence Theorem to our recursive formula, \eqref{eq:recursive}, and conclude that for any $t>0$ and any $k\in\N$, $\lim_{N\to\infty}\wh{F_{N,k}}\pa{n_1,\dots,n_k,t}=\wh{f_{k}}\pa{n_1,\dots,n_k,t}$ exists and satisfies\footnote{We need to be slightly careful here and employ an inductive argument to show that $f_{k-1}(t)$ is indeed a probability measure first. This is a very straightforward argument and as such we skip the details here.}
	\begin{equation}\label{eq:recursive_limit_in_fourier}
		\begin{aligned}
			&\wh{f_{k}}\pa{n_1,\dots,n_k,t} = e^{-\lambda k\pa{k-1}t}\wh{f_k}\pa{n_1,\dots,n_k}
			\\+&2\lambda\sum_{i<j\leq k}\int_{0}^t e^{-\lambda k\pa{k-1}\pa{t-s}} 
			\wh{f_{k-1}}\pa{n_1,\dots, n_i+n_j,\dots,n_k,s}ds,
		\end{aligned}
	\end{equation}
	where we have used the fact that $\lim_{N\to\infty}\g(n)=1$ for any $n\in\Z$. This shows \eqref{eq:limit_of_F_N_K_no_order} due to the uniqueness of the Fourier coefficients and the fact that
	$$\int_{\T^k} f_{k-1}\pa{\theta_1,\dots,\wt{\theta_{i}},\dots, \theta_k}\delta\pa{\theta_i-\theta_j}e^{-i \sum_{l=1}^k n_l\theta_l}\frac{d\theta_1\dots d\theta_k}{\pa{2\pi}^k}$$
	$$=\int_{\T^k} f_{k-1}\pa{\theta_1,\dots,\wt{\theta_{i}},\dots, \theta_k}e^{-i \sum_{l\not=i,\;l=1}^k n_l\theta_l}e^{-in_i\theta_j}\frac{d\theta_1\dots d\wt{\theta_i}\dots d\theta_k}{\pa{2\pi}^{k-1}}$$
	$$=\wh{f_{k-1}}\pa{n_1,\dots,n_i+n_j,\dots, n_k}.$$
	To show the convergence to an $f_1$-ordered state as time goes to infinity we notice that, just like the Lemma \ref{lem:weak_convergence_and_fourier} and by utilising Lemma \ref{lem:order_in_fourier}, it is enough for us to show that 
	$$\lim_{t\to\infty}\wh{f_k}\pa{n_1,\dots,n_k,t} = \wh{f_1}\pa{\sum_{j=1}^k n_j}.$$
	We will achieve this by showing that for any $k\geq 2$ there exists an explicit constant $c_k$ which depends only on $k$ such that 
	\begin{equation}\nonumber
		\abs{\wh{f_k}\pa{n_1,\dots,n_k,t}-\wh{f_1}\pa{\sum_{j=1}^k n_j}}\leq c_k e^{-2\lambda t}.
	\end{equation}
	We start by noticing that for $k=1$ \eqref{eq:recursive_limit_in_fourier} implies that 
	$$\wh{f_1}(n,t)=\wh{f_1}(n).$$
	Consequently, for $k=2$ we have that 
	\begin{equation}\label{eq:nice_eq_for_2}
		\begin{aligned}
			&\wh{f_{2}}\pa{n_1,n_2,t} = e^{-2\lambda t}\wh{f_2}\pa{n_1,n_2}+2\lambda\int_{0}^t e^{-2\lambda \pa{t-s}} 
			\wh{f_{1}}\pa{n_1+n_2,s}ds\\
			&= e^{-2\lambda t}\wh{f_2}\pa{n_1,n_2}+2\lambda\pa{\int_{0}^t e^{-2\lambda \pa{t-s}}ds} 
			\wh{f_{1}}\pa{n_1+n_2}\\
			&=e^{-2\lambda t}\wh{f_2}\pa{n_1,n_2} + \pa{1-e^{-2\lambda t}}\wh{f_1}\pa{n_1+n_2},
		\end{aligned}
	\end{equation}
	from which we find that 
	$$\abs{\wh{f_2}\pa{n_1,n_2,t}-\wh{f_1}\pa{n_1+n_2}} \leq 2 e^{-2\lambda t}=c_2e^{-2\lambda t}.$$
	We continue by induction: assume the claim holds for $k-1\geq 2$ and consider $k$. Since
	$$2\lambda \sum_{i<j\leq k}\int_0^t e^{-\lambda k\pa{k-1}\pa{t-s}}ds =\lambda k\pa{k-1}\int_0^t e^{-\lambda k\pa{k-1}\pa{t-s}}ds = 1- e^{-\lambda k\pa{k-1}t} $$
	we find that
	$$\abs{\wh{f_k}\pa{n_1,\dots,n_k,t} - \wh{f_1}\pa{\sum_{j=1}^k n_j}}\leq e^{-\lambda k\pa{k-1}t}\abs{\wh{f_1}\pa{\sum_{j=1}^k n_j}} + \Bigg| e^{-\lambda k\pa{k-1}t}\wh{f_k}\pa{n_1,\dots,n_k}$$
	$$
	+2\lambda\sum_{i<j\leq k}\int_{0}^t e^{-\lambda k\pa{k-1}\pa{t-s}} 
	\pa{\wh{f_{k-1}}\pa{n_1,\dots, n_i+n_j,\dots,n_k,s}-\wh{f_1}\pa{\sum_{j=1}^k n_j}}ds\Bigg|$$
	$$\leq 2 e^{-\lambda k\pa{k-1}t} + 2\lambda c_{k-1} \sum_{i<j\leq k}\int_{0}^t e^{-\lambda k\pa{k-1}\pa{t-s}} e^{-2\lambda s}ds$$
	$$=2 e^{-\lambda k\pa{k-1}t} + \lambda c_{k-1}k\pa{k-1} \frac{ e^{-2\lambda t}-e^{-\lambda k\pa{k-1}t} }{\lambda\pa{k\pa{k-1}-2}} $$
	where we have used \eqref{eq:useful_exponential} and the fact that $2\sum_{i<j\leq k}1 = k(k-1)$. Since $k\geq 3$ we see that 
	$$\abs{\wh{f_k}\pa{n_1,\dots,n_k,t} - \wh{f_1}\pa{\sum_{j=1}^k n_j}}\leq 2 e^{-\lambda k\pa{k-1}t} +  c_{k-1}k\pa{k-1} \frac{ e^{-2\lambda t}}{k\pa{k-1}-2} $$
	$$\leq \pa{2+\frac{c_{k-1}k\pa{k-1}}{k\pa{k-1}-2}}e^{-2\lambda t} = c_k e^{-2\lambda t}.$$
	We have thus shown the first statement of the theorem.\\ 
	Next, we show the propagation of order by induction. Recall that according to Lemma \ref{lem:purely_ordered_state} it will be enough for us to show that 
	$$\wh{f_k}\pa{n_1,\dots,n_k,t} = \wh{f_1}\pa{\sum_{j=1}^k n_j}$$
	for any $t>0$ and $\pa{n_1,\dots,n_k}\in\Z^k$.  Using Lemma \ref{lem:order_in_fourier} and the fact that $\br{F_{N}(0)}_{N\in\N}$ is $f_1-$ordered we conclude that 
	$$\wh{f_2}\pa{n_1,n_2}=\lim_{N\to\infty}\wh{F_{N,2}}\pa{n_1,n_2,0}=\wh{f_1}\pa{n_1+n_2}.$$
	Using the fact that $\wh{f_1}(n,t)=\wh{f_1}(n)$ for all $t>0$ together with the above and \eqref{eq:nice_eq_for_2} we find that
	\begin{equation}\nonumber
		\begin{aligned}
			&\wh{f_{2}}\pa{n_1,n_2,t} =e^{-2\lambda t}\wh{f_2}\pa{n_1,n_2} + \pa{1-e^{-2\lambda t}}\wh{f_1}\pa{n_1+n_2} = \wh{f_1}\pa{n_1+n_2},
		\end{aligned}
	\end{equation}
	which shows our base induction step. We now assume that 
	$$\wh{f_{k-1}}\pa{n_1,\dots,n_{k-1},t} = \wh{f_1}\pa{\sum_{l=1}^{k-1}n_l}$$ 
	for all $t>0$ and $\pa{n_1,\dots,n_{k-1}}\in\Z^{k-1}$, where $k-1\geq 2$. As in the case $k=2$ we know that the fact that $\br{F_{N}(0)}_{N\in\N}$ is $f_1$-ordered implies that 
	$$\wh{f_k}\pa{n_1,\dots,n_k}=\lim_{N\to\infty}\wh{F_{N,k}}\pa{n_1,\dots,n_k,0}=\wh{f_1}\pa{\sum_{j=1}^k  n_j}.$$
	Using our recursive formula \eqref{eq:recursive_limit_in_fourier} we find that for any $t>0$
	\begin{equation}\nonumber
		\begin{aligned}
			&\wh{f_{k}}\pa{n_1,\dots,n_k,t} = e^{-\lambda k\pa{k-1}t}\wh{f_k}\pa{n_1,\dots,n_k}
			\\+&2\lambda\sum_{i<j\leq k}\int_{0}^t e^{-\lambda k\pa{k-1}\pa{t-s}} 
			\wh{f_{k-1}}\pa{n_1,\dots, n_i+n_j,\dots,n_k,s}ds\\
			=&e^{-\lambda k\pa{k-1}t}\wh{f_{1}}\pa{\sum_{j=1}^k n_j}
			+2\lambda\sum_{i<j\leq k}\int_{0}^t e^{-\lambda k\pa{k-1}\pa{t-s}} 
			\wh{f_{1}}\pa{\sum_{j=1}^k n_j}ds
		\end{aligned}
	\end{equation}
	\begin{equation}\nonumber
		\begin{aligned}
			&=e^{-\lambda k\pa{k-1}t}\wh{f_{1}}\pa{\sum_{j=1}^k n_j}
			+2\lambda\sum_{i<j\leq k}\pa{\int_{0}^t e^{-\lambda k\pa{k-1}\pa{t-s}}ds} 
			\wh{f_{1}}\pa{\sum_{j=1}^k n_j}\\
			&=e^{-\lambda k\pa{k-1}t}\wh{f_{1}}\pa{\sum_{j=1}^k n_j}+\pa{1-e^{-\lambda k\pa{k-1}t}}\wh{f_{1}}\pa{\sum_{j=1}^k n_j}=\wh{f_{1}}\pa{\sum_{j=1}^k n_j}.
		\end{aligned}
	\end{equation}
	The proof, and with it this section, is now complete
\end{proof}

\section{The case of balanced interactions}\label{sec:balanced_interaction}

In this penultimate section, we consider the case where the interaction and time scaling are balanced. Surprisingly, Corollary \ref{cor:recursive} in \S\ref{sec:preliminaries} not only gives us the intuition to why Theorem \ref{thm:main_order} is true but also gives us the means to show that in the case where $N\epsilon_N^2=1$ the solutions to the rescaled CL model can't be ordered. The key idea in showing this is expressed in the following lemma:

\begin{lemma}\label{lem:necessary_for_order}
	Consider a family of symmetric probability densities $F_N\in \PP\pa{\T^N}$ with $N\in\N$. If $\br{F_N}_{N\in\N}$ is $f-$ordered then 
	$$\lim_{N\to\infty}\wh{F_{N,2}}\pa{n,-n} = 1.$$
\end{lemma}

\begin{proof}
	Using Lemma \ref{lem:order_in_fourier} we see that if $\br{F_N}_{N\in\N}$ is $f-$ordered then
	$$\lim_{N\to\infty}\wh{F_{N,2}}\pa{n,-n}= \wh{f}\pa{n+(-n)}=\wh{f}(0)=1$$
	as $f\in \PP\pa{\T}$. 
\end{proof}

\begin{proof}[Proof of Theorem \ref{thm:main_in_between}]
	The fact that $\br{F_N(t)}_{N\in\N}$ is not chaotic has been shown in the works of Carlen et al. \cite{CCDW2013,CDW2013}. To show the lack of order we start by noticing that in this setting, Lemma \ref{lem:approx_g} implies that for any $n\in \Z$
	$$\lim_{N\to\infty} N\pa{1-\g(n)} = \frac{m_2 n^2}{2}.$$
	Consequently, assuming that $\br{F_{N,1}(0)}_{N\in\N}$ and $\br{F_{N,2}(0)}_{N\in\N} $ converge weakly to $f_1$ and $f_2$ respectively and using \eqref{eq:expression_for_fourier_F_N_2}, we find that for any $t>0$, 
	\begin{equation}\nonumber
		\begin{aligned}
			&\lim_{N\to\infty}\wh{F_{N,2}}\pa{n_1,n_2,t} =e^{-\lambda \pa{\frac{m_2}{2}\pa{n_1^2+n_2^2}+2}t}\wh{f_{2}}\pa{n_1,n_2}\\
			&
			+\frac{4\pa{e^{-\lambda \pa{\frac{m_2}{2}\pa{n_1^2+n_2^2}+2}t}- e^{-\frac{\lambda m_2}{2}\pa{n_1+n_2}^2 t}}}{m_2\pa{\pa{n_1+n_2}^2-n_1^2-n_2^2} -4}\wh{f_{1}}\pa{n_1+n_2}
		\end{aligned}
	\end{equation}
	In particular, for any $n\not=0$ and $t>0$
	$$\lim_{N\to\infty}\wh{F_{N,2}}\pa{n,-n,t} =e^{-\lambda \pa{m_2 n^2+2}t}\wh{f_{2}}\pa{n,-n}-
	\frac{2\pa{e^{-\lambda \pa{m_2 n^2+2}t}- 1}}{m_2n^2+2} $$
	$$=e^{-\lambda \pa{m_2 n^2+2}t}\wh{f_{2}}\pa{n,-n}+
	\frac{2\pa{1-e^{-\lambda \pa{m_2 n^2+2}t}}}{m_2n^2+2}$$
	$$ \leq e^{-\lambda \pa{m_2 n^2+2}t} + \frac{2\pa{1-e^{-\lambda \pa{m_2 n^2+2}t}}}{m_2n^2+2} = \frac{m_2 n^2}{m_2n^2+2}e^{-\lambda \pa{m_2 n^2+2}t}  + \frac{2}{m_2n^2+2}$$
	$$< \frac{m_2 n^2}{m_2n^2+2}e^{-2 \lambda t}  + \frac{2}{m_2n^2+2} < 1.$$
	Due to Lemma \ref{lem:necessary_for_order} we conclude that $\br{F_N(t)}_{N\in\N}$ can't be ordered for any $t>0$, which completes the proof. 
\end{proof}

We conclude this short section with the proof of Theorem \ref{thm:not_all_hope_is_lost}.

\begin{proof}[Proof of Theorem \ref{thm:not_all_hope_is_lost}]
	Much like our previous proof, we start with the fact that in our setting
	$$\lim_{N\to\infty} N\pa{1-\g(n)} = \frac{m_2 n^2}{2}.$$
	Identity \eqref{eq:recursive_limit_in_fourier} together with the fact that $\br{F_{N,1}(0)}_{N\in\N}$ and $\br{F_{N,2}(0)}_{N\in\N} $ converge weakly to $f_1$ and $f_2$ respectively imply that
	$$\lim_{N\to\infty}\wh{F_{N,1}}\pa{n_1,t}=\lim_{N\to\infty}e^{\lambda N \pa{\wh{g_{\epsilon_N}}\pa{n}-1}t}\wh{F_{N,1}}\pa{n,0}=e^{-\frac{\lambda m_2}{2}n^2t}\wh{f_1}(n),$$
	and 
	\begin{equation}\nonumber
		\begin{aligned}
			&\lim_{N\to\infty}\wh{F_{N,2}}\pa{n_1,n_2,t} =e^{-\lambda \pa{\frac{m_2}{2}\pa{n_1^2+n_2^2}+2}t}\wh{f_{2}}\pa{n_1,n_2}\\
			&+
			\frac{4\pa{e^{-\lambda \pa{\frac{m_2}{2}\pa{n_1^2+n_2^2}+2}t}- e^{-\frac{\lambda m_2}{2}\pa{n_1+n_2}^2 t}}}{m_2\pa{\pa{n_1+n_2}^2-n_1^2-n_2^2} -4}\wh{f_{1}}\pa{n_1+n_2}.
		\end{aligned}
	\end{equation}
	The convergence of the Fourier coefficients together with Lemma \ref{lem:weak_convergence_and_fourier} imply the desired convergence to $f_1(t)$ and $f_2(t)$ given by the inverse transform of the above limits\footnote{More formally: if the Fourier coefficients of a given sequence of probability measures on $\T^k$ converge then the integration of that family against any trigonometric polynomial converges. As these polynomials are dense in $C_b\pa{\T^k}$ with respect to the uniform norm we conclude that the integration of that family against any bounded continuous function converges. This implies, according to the Riesz-Markov representation theorem on compact spaces, that the limit functional must be an integration against a probability measure whose Fourier coefficients are given by the limit of the Fourier coefficients of the original sequence.}. To show \eqref{eq:f_1_convergence} and \eqref{eq:f_2_convergence} we notice that 
	$$\lim_{t\to\infty}\wh{f_1}\pa{n_1,t} = \begin{cases}
		1,& n_1=0,\\
		0,& n_1\not=0,
	\end{cases} = \F_{\T}\pa{1}(n_1)$$
	and
	$$\lim_{t\to\infty}\wh{f_2}(n_1,n_2,t) = \begin{cases}
		\frac{4}{m_2 \pa{n_1^2+n_2^2}+4},& n_1+n_2=0,\\
		0,& n_1+n_2\not=0,
	\end{cases}= \begin{cases}
		\frac{2}{m_2 n_1^2+2},& n_1+n_2=0,\\
		0,& n_1+n_2\not=0,
	\end{cases}.$$
	The latter implies \eqref{eq:f_2_convergence} since (with the help of the Dominated Convergence Theorem) we have that
	$$\int_{\T^2} \pa{\sum_{j\in\Z}\frac{2}{m_2 j^2+2}e^{ij \pa{\theta_1-\theta_2}}}e^{-in_1\theta_1-in_2\theta_2}\frac{d\theta_1d\theta_2}{\pa{2\pi}^2}$$
	$$=\sum_{j\in\Z}\frac{2}{m_2 j^2+2}\int_{\T^2}e^{ij \pa{\theta_1-\theta_2}}e^{-in_1\theta_1-in_2\theta_2}\frac{d\theta_1d\theta_2}{\pa{2\pi}^2}=\sum_{j\in\Z}\frac{2\delta_{j,n_1}\delta_{j,-n_2}}{m_2 j^2+2}$$
	where $\delta_{i,j}$ is the Kronecker delta, and consequently
	$$\F_{\T^2}\pa{\mathcal{H}\pa{\cdot_1-\cdot_2}}\pa{n_1,n_2} = \begin{cases}
		\frac{2}{m_2 n_1^2+2},& n_1+n_2=0,\\
		0,& n_1+n_2\not=0.
	\end{cases}$$
\end{proof}

\section{Final remarks}\label{sec:final}

\subsection*{On the notion of order} Our definition of order (Definition \ref{def:order}) was motivated by our expectation to see \textit{total adherence} in the CL and other models -- a ``perpendicular'' behaviour to chaoticity. One might argue that a more appropriate name would be ``perfect order'' or ``perfect alignment'' to take into account that some partial order/alignment can also manifest (as might be indicated by Theorem \ref{thm:not_all_hope_is_lost}). However, to keep our introduction of this new asymptotic notion more coherent we elected to use the simpler term.\\ 
We would like to emphasise that the main idea behind the notion of order is that for any $k\in\N$ the limit process retains only \textit{one degree of randomness} (vs. chaoticity which has $k$ degrees of randomness). This means that this notion can be adapted to other situations where we don't necessarily expect that all the variables equal in the limit, but where one ``average element'' completely determines the limiting behaviour of any finite group of elements (for instance, a one dimensional chain of elements whose variables are always a fixed distance from each other). 
\subsection*{On the generation of order} As was mentioned in Remark \ref{rem:generation_of_order}, Theorem \ref{thm:main_order} guarantees the generation of order, though this statement is not as strong as we would hope. In particular, in order to see order appearing we need to consider the limiting marginals (i.e. take $N$ to infinity) and \textit{then} take time to infinity. It would be interesting to see if we can find an explicit function $t(N)$, that goes to infinity when $N$ goes to infinity, such that $F_{N,k}\pa{t\pa{N}}$ converges to an ordered state as $N$ goes to infinity. We suspect that to achieve this one might need a stronger notion of convergence than weak convergence of measures which is also \textit{quantitative}. 
\subsection*{Between order and chaos} The balanced setting, discussed in Theorems \ref{thm:main_in_between} and \ref{thm:not_all_hope_is_lost}, poses an interesting ``in between'' case between our order and suspected chaos. While no order is observed in this case, Theorem \ref{thm:not_all_hope_is_lost} suggests that there is still a chance we will see some partial adherence, at least in the second marginal, with deviations given by a fixed function. This motivates us to consider a potential notion of \textit{partial order}, where the delta functionals in \eqref{eq:def_order_simplified} are replaced by some functions that measure how close the variables may get. In other words $\Pi_1\pa{d\mu_N}$ converges to a profile $f$ and $\Pi_k\pa{d\mu_N}$ converges to something of the form
$$ \frac{1}{k!}\sum_{\sigma\in S_k}f\pa{\theta_{\sigma(1)}}\prod_{i=1}^{k-1} h\pa{\theta_{\sigma(i)}-\theta_{\sigma(i+1)}},$$
for some $h\in \PP\pa{\T}$ and where $S_k$ is the group of permutation of order $k$. It is unclear at this point if the above is suitable to capture the behaviour of even the simple CL model in the balanced scaling, but the investigation of such a notion is, in our opinion, an exciting prospect which we will pursue.
\subsection*{Additional models} The CL model did not only motivate the definition of the new notion of order -- it was also an ideal model to test it. One notable issue with this model, however, is its simplicity. In particular, its BBGKY hierarchy is closed -- something that doesn't happen in most many element models. It would be interesting to try and test the notion of order in other mean field models that should exhibit strong adherence. Prime candidates are swarming models such as the Bertin, Droz and Gr\'{e}goire model, which was introduced in \cite{BDG2006} and is mentioned in the works of Carlen et al \cite{CCDW2013,CDW2013}, and societal models such as the opinion models presented in the review paper of Chaintron and Diez, \cite{CD(II)2022}. Following on ideas presented in the original works on the CL model as well as in this paper, one would expect that the first step to deal with any mean field model which may exhibit a phenomena of order would be to find the appropriate scaling. This might not be as easy a feat as it is in the CL model and additional technical difficulties are expected due to the coupled BBGKY hierarchy.

\section*{Acknowledgement}
The author would like to thank the reviewers of this work for their careful reading and the numerous helpful suggestions they have given -- improving both the clarity and readability of the paper. 

\begin{appendix}
	
	\section{The behaviour of Fourier coefficients of rescaled and restricted probability densities}\label{app:fourier_behaviour}
	In this appendix we will prove Lemma \ref{lem:approx_g} by stating and proving two auxiliary lemmas.

	\begin{lemma}\label{lem:fourier_connection}
		Let $g\in \PP\pa{\R,dx}$. We define its \textit{$\epsilon$-truncated Fourier transform induced from $\pa{\T , \frac{d\theta}{2\pi}}$}, $\F_{\epsilon}\pa{g}$, to be the function
		$$\F_{\epsilon}(g)\pa{\xi} = \int_{-\frac{\pi}{\epsilon}}^{\frac{\pi}{\epsilon}}g(x)e^{-i\xi x}dx.$$
		Then for any $n\in\Z$
		\begin{equation}\nonumber 
			\widehat{g_{\epsilon}}(n) = \frac{\F_{\epsilon}(g)\pa{n\epsilon}}{\F_{\epsilon}(g)\pa{0}}.
		\end{equation}
		Moreover, if there exists $k\in\N$ such that 
		$$m_k=\int_{\R}\abs{x}^k g(x)dx<\infty$$
		then for any $n\in\Z$ we have that
		\begin{equation}\nonumber 
			\abs{\widehat{g_{\epsilon}}(n)-\F(g)\pa{n\epsilon}} \leq \frac{2\epsilon^k m_k}{\pi^k-\epsilon^km_k}
		\end{equation}
		whenever $\epsilon  < \frac{\pi}{\sqrt[k]{m_k}}$ and where $\F\pa{g}$ is the Fourier transform of $g$
		$$\F(g)\pa{\xi} = \int_{\R}g(x)e^{-i\xi x}dx.$$
	\end{lemma}
	
	\begin{proof}
		By the definition of $g_{\epsilon}$ and using the fact that 
		$$\widetilde{g}_\epsilon = \frac{1}{2\pi}\F_{\epsilon}(g)\pa{0}$$
		we have that
		$$\widehat{g_{\epsilon}}(n) = \frac{1}{2\pi \epsilon\widetilde{g}_\epsilon}\int_{-\pi}^{\pi}g\pa{\frac{\theta}{\epsilon}}e^{-in\theta}d\theta
		= \frac{1}{\F_{\epsilon}(g)\pa{0}}\int_{-\frac{\pi}{\epsilon}}^{\frac{\pi}{\epsilon}}g(x)e^{-in \epsilon x}dx=\frac{\F_{\epsilon}(g)\pa{n\epsilon}}{\F_{\epsilon}(g)\pa{0}},$$
		which gives us the first claim. To show the second claim we start by noticing that 
		$$\abs{\F_{\epsilon}(g)\pa{\xi}-\F(g)\pa{\xi}} \leq \int_{\abs{x} > \frac{\pi}{\epsilon}}g(x)dx \leq \frac{\epsilon^k}{\pi^k}m_k.$$
		Consequently, since $\F\pa{g}(0)=1$, we have that if $\epsilon < \frac{\pi}{\sqrt[k]{m_k}}$
		$$\abs{\widehat{g_{\epsilon}}(n)-\F(g)\pa{n\epsilon}}=\abs{\frac{\F_{\epsilon}(g)\pa{n\epsilon}}{\F_{\epsilon}(g)\pa{0}}-\F(g)\pa{n\epsilon}}$$
		$$\leq \frac{1} {1-\abs{\F_{\epsilon}(g)\pa{0}-1}}\pa{\abs{\F_{\epsilon}(g)\pa{n\epsilon}-\F(g)\pa{n\epsilon}}+\abs{\F_{\epsilon}(g)\pa{0}-1}\abs{\F(g)\pa{n\epsilon}}}$$
		$$\leq \frac{2\frac{\epsilon^k}{\pi^k}m_k}{1-\frac{\epsilon^k}{\pi^k}m_k}=\frac{2\epsilon^k m_k}{\pi^k-\epsilon^km_k}.$$
		The proof is now complete. 
	\end{proof}

	\begin{lemma}\label{lem:some_fourier_trasnform_estimation}
		Let $g\in\PP\pa{\R,dx}$ be a symmetric probability density such that $m_k$, defined in the above lemma, is finite for some $k\in\N$ such that $k>2$. Then 
		\begin{equation}\nonumber 
			\abs{\F(g)\pa{\xi} - 1 +\frac{m_2 \xi^2}{2}} \leq \begin{cases}
				\frac{m_3}{3}\abs{\xi}^3, & k=3,\\
				\frac{m_4}{12}\abs{\xi}^4, & k>3,
			\end{cases}
			\qquad \forall\xi\in\R.
		\end{equation}
	\end{lemma}
	
	\begin{proof}
		From the definition of $m_2$ and the fact that $g$ is symmetric we find that 
		$$\abs{\F(g)\pa{\xi} - 1 + \frac{m_2 \xi^2}{2}} = \abs{\int_{\R}g(x)\pa{e^{-i\xi x} - 1 - i\xi x -\frac{\pa{i\xi}^2 x^2}{2!}}dx}$$
		$$\leq \int_{\R}g(x)\abs{\cos\pa{\xi x} - 1 +\frac{\xi^2 x^2}{2}}dx + \int_{\R}g(x)\abs{\sin\pa{\xi x} -\xi x}dx.$$
		Since
		$$\max\br{\abs{\cos\pa{t}-1-\frac{t^2}{2}},\abs{\sin\pa{t}-t}} \leq \frac{\abs{t}^3}{3!}$$
		we see that if $m_3<\infty$ then 
		\begin{equation}\nonumber
			\abs{\F(g)\pa{\xi} - 1 + \frac{m_2 \xi^2}{2}}\leq \frac{m_3}{3}\abs{\xi}^3. 
		\end{equation}
		If in addition we have that $m_4<\infty$ then, since
		$$\max\br{\abs{\cos\pa{t}-1-\frac{t^2}{2}},\abs{\sin\pa{t}-t+ \frac{t^3}{3!}}} \leq \frac{\abs{t}^4}{4!}$$
		and since 
		$$\int_{\R}x^3g(x)dx=0$$
		we can refine the above estimate to find that
		$$\abs{\F(g)\pa{\xi} - 1 + \frac{m_2 \xi^2}{2}} = \abs{\int_{\R}g(x)\pa{e^{-i\xi x} - 1 - i\xi x -\frac{\pa{i\xi}^2 x^2}{2!} -\frac{\pa{i\xi}^3 x^3}{3!}}dx}$$
		$$\leq \int_{\R}g(x)\abs{\cos\pa{\xi x} - 1 -\frac{\xi^2 x^2}{2}}dx + \int_{\R}g(x)\abs{\sin\pa{\xi x} - \xi x + \frac{\xi^3 x^3}{3!}}dx\leq \frac{m_4}{12}\abs{\xi}^4,$$
		which concludes the proof.
	\end{proof}
	
	\begin{proof}[Proof of Lemma \ref{lem:approx_g}]
		The proof is an immediate consequence of lemmas \ref{lem:fourier_connection} and  \ref{lem:some_fourier_trasnform_estimation}.	
	\end{proof}
	
	\section{Additional proofs}\label{app:extra}
	In this short appendix we will show a claim that was stated in the proof of Lemma  \ref{lem:fourier_of_F_N_k} and prove Lemma \ref{lem:weak_convergence_and_fourier}.

	\begin{lemma}\label{lem:symmtery_in_fourier}
		Let $f\in \PP\pa{\T^k}$. Then if $f$ is symmetric so it $\wh{f}$. 
	\end{lemma}
	\begin{proof}
		For any permutation $\sigma\in S_{k}$ we have that
		$$\wh{f}\pa{n_1,\dots,n_k} = \int_{\T^k} f\pa{\theta_1,\dots,\theta_k}e^{-i\sum_{i=1}^k n_i \theta_i}\frac{d\theta_1\dots d\theta_k}{\pa{2\pi}^k}$$
		$$ = \int_{\T^k} f\pa{\theta_1,\dots,\theta_k}e^{-i\sum_{i=1}^k n_{\sigma(i)} \theta_{\sigma(i)}}\frac{d\theta_1\dots d\theta_k}{\pa{2\pi}^k}$$
		$$= \int_{\T^k} f\pa{\theta_{\sigma(1)},\dots,\theta_{\sigma(k)}}e^{-i\sum_{i=1}^k n_{\sigma(i)} \theta_{\sigma(i)}}\frac{d\theta_1\dots d\theta_k}{\pa{2\pi}^k}$$
		$$= \int_{\T^k} f\pa{\theta_{1},\dots,\theta_{k}}e^{-i\sum_{i=1}^k n_{\sigma(i)} \theta_{i}}\frac{d\theta_1\dots d\theta_k}{\pa{2\pi}^k} = \wh{f}\pa{n_{\sigma(1)},\dots,n_{\sigma(k)}} $$
		which concludes the proof. 
	\end{proof}
	
	\begin{proof}[Proof of Lemma \ref{lem:weak_convergence_and_fourier}]
			Since $e^{-i\sum_{j=1}^k n_j \theta_j} \in C_b\pa{\T^k}$, one implication is immediate. The converse is also straight forward and follows easily from the fact that trigonometric polynomials are dense in $C_b\pa{\T}$ with respect to the uniform norm\footnote{Recall that $\T$ stands for $\Sp^1$ which means that our continuous functions are $2\pi-$periodic.}. 
	\end{proof}
	
\end{appendix}




\begin{thebibliography}{99}
	
	\bibitem{APD2021}
	Ayi, N and Pouradier Duteil, N., \textit{Mean-field and graph limits for collective dynamics models
		with time-varying weights}, J. Differential Equations \textbf{299} (2021). https://doi.org/10.1016/j.jde.2021.07.010
	
	\bibitem{BCC2011}
	Bolley, F. and Ca\~{n}izo, J. A. and Carrillo, J. A., \textit{Stochastic mean-field limit: non-{L}ipschitz forces and
		swarming}, Math. Models Methods Appl. Sci. \textbf{11} (2011). https://doi.org/10.1142/S0218202511005702
	
	\bibitem{BDG2006}
	Bertin, E. and Droz, M, and Gr\'egoire, G., \textit{Boltzmann and hydrodynamic description for self-propelled particles}, Phys. Rev. E \textbf{74} (2006). https://doi.org/10.1103/PhysRevE.74.022101
	
	\bibitem{BFFT2012}
	Baladron, J. and Fasoli, D. and Faugeras, O. and Touboul J., \textit{Mean-field description and propagation of chaos in networks of
		{H}odgkin-{H}uxley and {F}itz{H}ugh-{N}agumo neurons}, J. Math. Neurosci. \textbf{2} (2012). https://doi.org/10.1186/2190-8567-2-10
	
	\bibitem{CCDW2013}
	Carlen, E. A. and Chatelin, R. and Degond, P. and Wennberg, B., \textit{Kinetic hierarchy and propagation of chaos in biological swarm models}, Phys. D \textbf{260} (2013), 90--111. https://doi.org/10.1016/j.physd.2012.05.013
	
	\bibitem{CDW2013}
	Carlen, E. A. and Degond, P. and Wennberg, B., \textit{Kinetic limits for pair-interaction driven master equations
		and biological swarm models}, Math. Models Methods Appl. Sci. \textbf{23} (2013), 1339--1376. https://doi.org/10.1142/S0218202513500115
	
	\bibitem{CD(II)2022}
	Chaintron, L.-P. and Diez, A., \textit{Propagation of chaos: a review of models, methods and
		applications. II. Applications}, Kinet. Relat. Models \textbf{15} (2022), 1017--1173. https://doi.org/10.3934/krm.2022018
	
	
	
	\bibitem{GSRT2013}
	Gallagher, I. and Saint-Raymond, L. and Texier, B., \textit{From Newton to Boltzmann: hard spheres and short-range
		potentials}, Zurich Lectures in Advanced Mathematics (2013). https://doi.org/10.4171/129
	
	
	\bibitem{HM2014}
	Hauray, M. and Mischler, S., \textit{On {K}ac's chaos and related problems}, J. Funct. Anal. \textbf{10} (2014). https://doi.org/10.1016/j.jfa.2014.02.030
	
	
	\bibitem{K1956}
	Kac, M., \textit{Foundations of kinetic theory}, Proceedings of the Third Berkeley Symposium on
	Mathematical Statistics and Probability, 1954--1955,
	vol. III (1956), 171--197.
	
	\bibitem{Lanford1975}
	Lanford, III, O. E., \textit{Time evolution of large classical systems}, Dynamical systems, theory and applications \textbf{38} (1975), 1--111. https://doi.org/10.1007/3-540-07171-7\_1 
	
	\bibitem{Mckean1967}
	McKean, Jr., H. P., \textit{An exponential formula for solving Boltzmann's equation for a
		Maxwellian gas}, J. Combinatorial Theory \textbf{2} (1967), 358--382. https://doi.org/10.1016/S0021-9800(67)80035-8
	
	\bibitem{RS2023}
	Matthew Rosenzweig and Sylvia Serfaty. Modulated logarithmic Sobolev inequalities and generation of chaos. arXiv preprint: arXiv:2307.07587 (2023).
\end{thebibliography}
\end{document}